\documentclass[12pt]{article}

\input{style}

\usepackage{amsmath, amssymb, amstext, amsfonts, mathrsfs}

\begin{document}

\title{\vspace*{-1.5cm} Hedging with Temporary Price Impact}

\author{
 Peter Bank\footnote{Technische Universit{\"a}t Berlin, Institut f{\"u}r Mathematik, 
 Stra{\ss}e des 17. Juni 136, 10623 Berlin, Germany, email \texttt{bank@math.tu-berlin.de}. 
 Financial support by Einstein Foundation through project ``Game options
 and markets with frictions'' is gratefully acknowledged.} 
 \hspace{2ex} 
 H. Mete Soner\footnote{ETH Z\"urich, Departement f\"ur Mathematik, R\"amistrasse 101, CH-8092, Z\"urich,
 Switzerland, and Swiss Finance Institute, email \texttt{mete.soner@math.ethz.ch}.} 
 \hspace{2ex} 
 Moritz Vo{\ss}\footnote{Technische Universit{\"a}t Berlin, Institut f{\"u}r Mathematik, 
 Stra{\ss}e des 17. Juni 136, 10623 Berlin, Germany, email \texttt{voss@math.tu-berlin.de}.}
}
\date{\today}

\maketitle
\begin{abstract}
  We consider the problem of hedging a European contingent claim in a
  Bachelier model with temporary price impact as proposed
  by~\citet{AlmgChr:01}. Following the approach of~\citet{RogerSin:10}
  and~\citet{NaujWes:11}, the hedging problem can be regarded as a
  cost optimal tracking problem of the frictionless hedging strategy.
  We solve this problem explicitly for general predictable target
  hedging strategies. It turns out that, rather than towards the
  current target position, the optimal policy trades towards a
  weighted average of expected future target positions. This
  generalizes an observation of~\citet{GarlPeder:13.2} from their
  homogenous Markovian optimal investment problem to a general hedging
  problem. Our findings complement a number of previous studies in the
  literature on optimal strategies in illiquid markets as, e.g.,
  \cite{GarlPeder:13.2}, \cite{NaujWes:11}, \cite{RogerSin:10},
  \cite{AlmgLi:15}, \cite{MorMKSoner:15}, \cite{KallMK:14},
  \cite{GuasoniWeb:15_1}, \cite{GuasoniWeb:15_2}, where the
  frictionless hedging strategy is confined to diffusions. The
  consideration of general predictable reference strategies is made
  possible by the use of a convex analysis approach instead of the
  more common dynamic programming methods.
\end{abstract}

\begin{description}
\item[Mathematical Subject Classification (2010):] 91G10, 91G80,\\
  91B06, 60H30
\item[JEL Classification:] G11, C61
\item[Keywords:] Hedging, illiquid markets, portfolio tracking
\end{description}

\newpage
\section{Introduction}

The construction of effective hedging strategies against financial
risk is one of the key problems in Mathematical Finance. The seminal
work of~\citet{BlacSch:73} and~\citet{Mert:73} showed how this task
can be carried out in an idealized frictionless market by dynamically
trading perfectly liquid assets. However, in recent years there has
been a growing awareness that these idealizations can lead to
misguided hedging strategies with non negligible costs, particularly
when these prescribe a fast reallocation of assets in short periods of
time in the presence of liquidity frictions like price impact. This
has spurred the development of new financial models which take into
account the impact of transactions on execution prices; see, e.g., the
survey by~\citet{GokayRochSoner:11}.

The two most widely used models go back to~\citet{AlmgChr:01} as well
as~\citet{ObiWang:13}, respectively: Loosely speaking, the model of
Almgren and Chriss is characterized by directly specifying functions
describing the temporary and permanent impacts of a given order on the
price. The model of Obizhaeva and Wang assumes that trading takes
place in a block-shaped limit order book with persistent price impact
which is vanishing at a finite resilience rate. As recently discussed
in~\citet{KallMK:14}, the former can be regarded as the
high-resilience limit of the latter.
 
Within these two models, most of the existing literature studies the
problem of optimally liquidating an exogenously given position by some
fixed time horizon (cf., e.g., \citet{AlmgChr:01}, \citet{Almg:03},
\citet{SchSchon:09}, \citet{ObiWang:13}, \citet{AlfFruSch:10} and
\citet{PredShaShr:11}). Further work is also devoted to the more
involved problem of optimal portfolio choice, cf., e.g.,
\citet{GarlPeder:13.1}, \cite{GarlPeder:13.2}, \citet{MorMKSoner:15},
\citet{GuasoniWeb:15_1}, \cite{GuasoniWeb:15_2} and \citet{KallMK:14}.
However, only a few papers directly address the problem of hedging a
contingent claim in the presence of price impact as modeled above,
cf. \citet{RogerSin:10}, \citet{AlmgLi:15}, \citet{GueantPu:15}, and
also \citet{NaujWes:11}.
  
The papers most closely related to ours mathematically are
\citet{RogerSin:10} and \citet{NaujWes:11}. \citet{RogerSin:10}
analyse the problem of hedging a European contingent claim in a
Black-Scholes model in the presence of purely temporary price impact
as in~\citet{AlmgChr:01}.  They relate the hedging problem to a cost
optimal tracking problem of the frictionless Black-Scholes delta
hedge. \citet{NaujWes:11} directly study the problem of optimally
following a given target strategy in an illiquid financial market
under the same type of liquidity costs; see also \citet{CarJai:15} for
a Markovian order flow tracking problem. By contrast to these papers,
we will focus on a non-Markovian setup with general predictable target
strategies.

Instead of the more common dynamic programming methods used in the
papers cited above, our approach is convex analytic along the lines of
Pontryagin's maximum principle.  This allows us to consider general
predictable target strategies and not only continuous diffusion-type
processes. This is particularly important for hedging in illiquid
markets when the frictionless reference hedge portfolio prescribes
sizable instantaneous reallocations as, e.g., already in the case of
discrete Asian options which was not covered by the literature so far.
We derive first order conditions of the considered quadratic
optimization problems which take the form of a linear forward backward
stochastic differential equation (FBSDE). Solutions to these are
explicitly available and give us the optimal frictional hedges.  In
fact, when considered in a Brownian setting, our approach can be
viewed as a special case of the stochastic linear-quadratic control
problem studied by \citet{KohlmannTang:02}. Mathematically, the
novelty of our contribution is the interpretation of the optimal
tracking strategy. Indeed, it turns out that the optimal policy does
\emph{not} instantaneously trade from its current position towards the
current target position but towards a weighted average of
\emph{expected future target positions} which does not occur in the
work of \citet{KohlmannTang:02}. An interesting consequence from a
financial point of view is that this averaging allows us to understand
how singularities in the frictionless reference strategy have to be
addressed in a model with frictions: A singularity in the frictionless
target hedge is smoothed out when averaging the weighted future target
positions which yields sensible hedging strategies for illiquid
markets. Additionally, we also study a constrained version of the
problem where the terminal hedging position is restricted to a certain
exogenously prescribed level. This can be viewed as a way to deal with
physical delivery in derivative contracts. Our explicit solution
reveals how the hedger's focus shifts systematically from tracking the
frictionless target position to attaining the prescribed terminal
position. Here, our convex analytic approach allows us to avoid the
consideration of nonlinear Hamilton-Jacobi-Bellmann equations with
singular terminal conditions and the challenges that these entail. We
also give a sharp characterization of those terminal positions which
can be reached with finite expected trading costs by characterizing
the speed at which the size of these positions is revealed towards the
end.

Conceptually, our result generalizes an observation
by~\citet{GarlPeder:13.2} who consider quadratic utility maximization
in homogeneous Markovian models on an infinite time horizon and
interpret their solution as trading towards an exponentially weighted
average of future expected Markowitz portfolios. A similar
interpretation is given by \citet{NaujWes:11} in their equally
Markovian Example~7.1; see \citet{CarJai:15} for a similarly Markovian
study of tracking of order flow in high-frequency trading. These
strategies as well as ours contrast with strategies targeting the
present frictionless optimum directly, which are considered in many
papers on asymptotically optimal portfolios with small transaction
costs, including \citet{RogerSin:10}, \citet{MorMKSoner:15},
\citet{GuasoniWeb:15_1}, \cite{GuasoniWeb:15_2}, and
\citet{KallMK:14}. In all the literature cited above, the authors
confine consideration to diffusion-type target strategies which, at
least asymptotically, are equivalent to our averaged targets. Our
approach, by contrast, allows one to deal with general predictable
frictionless target strategies and so the examples considered in this
paper include strategies with jumps or even singularities where the
differences between these hedges become apparent.

\citet{AlmgLi:15} study a quite similar hedging problem but they
consider a model with permanent price impact which feeds into their
target strategies via the well-known functions for Black-Scholes
deltas and gammas. Hence, they consider a model where the target
strategy is also affected by the targeting strategy which leads to a
feedback effect that we are disregarding in our problem
formulation. We refer to the introduction in~\citet{RogerSin:10} for
further discussion of this idealization.

The rest of the paper is organized as follows. In
Section~\ref{sec:problem} we introduce the setup and motivate our
problem formulation by following the approach of Rogers and
Singh~\cite{RogerSin:10} (cf. also~\citet{NaujWes:11}). Our main
result is presented in Section~\ref{sec:main} and provides the
explicit solution for a general hedging problem of a European style
option in a Bachelier market model with temporary price impact.
Section~\ref{sec:illustrations} contains some illustrations of optimal
solutions in three examples.  The technical proofs are deferred to
Section \ref{sec:proofs}.


\section{Problem setup and motivation}
\label{sec:problem}

We fix a finite deterministic time horizon $T > 0$, a filtered
probability space $(\Omega,\mathcal{F},(\cF_t)_{0 \leq t \leq T},\PP)$
satisfying the usual conditions of right continuity and completeness
and consider an agent who is trading in a financial market consisting
of a risky asset, e.g., stock. The number of shares the agent holds at
time $t \in [0,T]$ of the risky stock is defined as
\begin{equation}
  \label{eq:portfolio} 
   X^u_t \set x + \int_0^t u_s ds , \quad 0 \leq t \leq T,
\end{equation}
where $x \in \RR$ denotes her given initial holdings. The real-valued
stochastic process $(u_t)_{0 \leq t \leq T}$ represents the agent's
turnover rate, that is, the speed at which the agent trades in the
risky asset. It is assumed to be chosen in the set
\begin{equation*}
  \cU \set \left\{ u : u \text{ progressively measurable s.t. } 
    \mathbb{E} \int_0^T u^2_t dt < \infty \right\}.
\end{equation*}
The square-integrability requirement ensures that the induced
quadratic transaction costs which are levied on the agent's respective
turnover rates due to temporary price impact as in \citet{AlmgChr:01}
are finite.

In such a frictional market, our agent seeks to track a target
strategy which can be thought of, for instance, as a hedging strategy
adopted from a frictionless setting. Mathematically, this problem can
be formalized as follows: Given a real-valued predictable process
$(\xi_t)_{0 \leq t \leq T}$ in $L^2(\PP \otimes dt)$ and a fixed
constant $\kappa > 0$, the agent's objective is to minimize the
performance functional
\begin{equation} \label{eq:objFun}
   J(u) \set \EE \left[ \frac{1}{2} \int_0^T (X^u_t - \xi_t)^2 dt 
   + \frac{1}{2} \kappa \int_0^T u^2_t dt \right].
\end{equation}
This leads to the optimal stochastic control problem
\begin{equation} \label{eq:optProb1}
   J(u) \rightarrow \min_{u \in \cU}. 
\end{equation}
Since the agent's terminal position $X^u_T$ may be important (for here
future plans or physical delivery), we also consider the optimal
stochastic control problem
\begin{equation} \label{eq:optProb2}
   J(u) \rightarrow \min_{u \in \cU^{\Xi}_x}
\end{equation}      
where $\cU^{\Xi}_x$ is the set of constrained policies defined as
\begin{equation*}
  \cU^{\Xi}_x \set \left\{ u : u \in \cU \text{ satisfying } 
    X^u_T \equiv x + \int_0^T u_s ds = \Xi_T \,\, \PP\text{-a.s.} \right\}
\end{equation*}
with predetermined terminal position $\Xi_T \in L^2(\PP,\cF_T)$ such
that
\begin{equation} \label{eq:ccp}
\int_0^T \frac{d\EE[\Xi_t^2]}{T-t} < \infty
\end{equation}
where $\Xi_t \set \EE[\Xi_T \vert \cF_t]$ for $0 \leq t \leq T$.

\begin{Remark}
  \begin{itemize}
  \item[1.)]  Lemma \ref{lem:ccp} below shows that a target $\Xi_T$
    can be reached with finite expected costs in the sense that
    $\cU^{\Xi}_x \neq \varnothing$ if and only if \eqref{eq:ccp} is
    satisfied.  Observe that this condition implies, in particular,
    that $\Xi_T \in \cF_{T-}$.  In fact, \eqref{eq:ccp} can be
    interpreted as a condition on the speed at which one learns about
    the ultimate target position $\Xi_T$ as $t \uparrow T$.
  \item[2.)] Concerning physical delivery at maturity $T$, it would be
    sufficient to impose the constraint $X_T^u \geq \Xi_T$.  However,
    this would lead to an interesting, yet technically rather
    different optimization problem which is left for future research.
  \end{itemize}
\end{Remark}

One motivation of the objective functional in (\ref{eq:objFun}) and
its connection to the problem of hedging a European contingent claim
in the presence of temporary price impact is the following
(cf. also~\citet{RogerSin:10} and~\citet{AlmgLi:15}): Assume the agent
wants to hedge a European-type option with payoff $H$ at maturity $T$
in a market where, for simplicity, interest rates are zero and the
price process $S$ of the underlying risky asset follows a Brownian
motion with volatility $\sigma > 0$:
\begin{equation*}
  S_t = S_0 + \sigma W_t , \quad 0 \leq t \leq T.
\end{equation*}
In a frictionless setting, the payoff $H$ can be perfectly replicated
by a predictable hedging strategy $\xi^H$.  In a market with frictions
where the agent faces liquidity costs as, e.g., in~\citet{AlmgChr:01},
she may be confined to follow strategies $X^u$ as in
(\ref{eq:portfolio}). As a consequence, starting from some initial
wealth $v_0 \in \RR$, her profits and losses from market fluctuations
will incur a risk of deviating from $H$ at maturity $T$ that can be
measured, e.g., by
\begin{equation*}
  \EE\left[ \left( H - (v_0+\int_0^T X^u_t dS_t) \right)^2 \right] = 
  (\EE[H]-v_0)^2 + \EE\left[ \int_0^T (X_t^u - \xi^H_t )^2 \sigma^2 dt \right],
\end{equation*}
see~\citet{FollSond:86}. This deviation can be made arbitrarily small
if the agent is willing to incur arbitrarily high transaction
costs. If, however, she puts a cap $c > 0$ on these she may want to
solve the optimization problem
\begin{equation}
\label{eq:motiv}
\EE\left[ \int_0^T (X_t^u - \xi^H_t )^2 \sigma^2 dt \right] \rightarrow \min_{u \in \cU}
\quad \text{subject to } \mathbb{E} \left[ \int_0^T u^2_t dt \right] \leq c,
\end{equation}
which in its Lagrangian formulation amounts to an objective functional
of the form (\ref{eq:objFun}).

\begin{Remark} 
\label{rem:prob}
\begin{enumerate}
\item A similar hedging problem as formulated in (\ref{eq:objFun}) is
  also studied in \citet{RogerSin:10} and \citet{AlmgLi:15}. In
  contrast to our setting, \citet{RogerSin:10} consider a
  Black-Scholes framework.~\citet{AlmgLi:15} also include permanent
  price impact.

\item Apart from hedging, the minimization problem of the objective in
  (\ref{eq:objFun}) is also related to the problem of optimally
  executing a VWAP order as studied using dynamic programming methods
  in a Markovian setup in \citet{FreiWes:13} and \citet{CarJai:15},
  or, more generally, to the optimal curve following problem as
  discussed in \citet{NaujWes:11} as well as
  \citet{CaiRosenbaumTankov:15}.

\item In a Brownian setting, our problem \eqref{eq:optProb1} is a
  special case of a stochastic linear-quadratic control problem as
  studied, e.g., by \citet{KohlmannTang:02}.
\end{enumerate}
\end{Remark}


\section{Main result}
\label{sec:main}

Our main results are the following explicit descriptions of the
optimal controls for problems (\ref{eq:optProb1}) and
(\ref{eq:optProb2}) and their corresponding minimal costs for which it
is convenient to introduce
$$
\tau^\kappa(t) \set \frac{T-t}{\sqrt{\kappa}} , \quad 0 \leq t \leq T.
$$
\begin{Theorem} \label{thm:main1} The optimal stock holdings $\hat{X}$
  of problem (\ref{eq:optProb1}) with unconstrained terminal position
  satisfy the linear ODE
   \begin{equation} \label{eq:SDEX1}
      d\hat{X}_t = \frac{\tanh (\tau^\kappa(t))}{\sqrt{\kappa}}
      \left(\hat{\xi}_t - \hat{X}_t \right) dt, \quad \hat{X}_0=x,
   \end{equation}
   where, for $0 \leq t < T$, we let
   \begin{equation*}
     \hat{\xi}_t \set
     \EE \left[ \int_t^T \xi_u K(t,u) du \bigg\vert \mathcal{F}_t \right] \quad 
   \end{equation*}
   with the kernel 
   \begin{equation*}
      K(t,u) \set \frac{\cosh(\tau^\kappa(u))}
      {\sqrt{\kappa}\sinh(\tau^\kappa(t))}, \quad 0 \leq t \leq u < T.  
    \end{equation*}
    The minimal costs are given by
    \begin{align} 
      \inf_{u \in \cU} J(u) = 
      & \frac{1}{2} \sqrt{\kappa} \tanh(\tau^\kappa(0)) \left( x - \hat{\xi}_0 \right)^2 +
        \frac{1}{2} \EE \left[ \int_0^T (\xi_t - \hat{\xi}_t)^2 dt \right] \nonumber \\
      & + \frac{1}{2} \EE \left[ \int_0^T \sqrt{\kappa}
        \tanh(\tau^\kappa(t)) d\langle \hat{\xi} \rangle_t \right]
        < \infty. 
    \label{eq:cost1}
    \end{align}
\end{Theorem}

For the constrained problem we have similarly:

\begin{Theorem} \label{thm:main2} The optimal stock holdings
  $\hat{X}^{\Xi}$ of problem (\ref{eq:optProb2}) with constrained
  terminal position $\Xi_T \in L^2(\PP,\cF_{T})$ such that
  \eqref{eq:ccp} holds satisfy the linear ODE
   \begin{equation} \label{eq:SDEX2}
      d\hat{X}^{\Xi}_t = \frac{\coth (\tau^\kappa(t))}{\sqrt{\kappa}} 
       \left(\hat{\xi}^{\Xi}_t - \hat{X}^{\Xi}_t \right) dt,
      \quad \hat{X}^{\Xi}_0 = x,
   \end{equation}
   where, for $0 \leq t \leq T$, we let
   \begin{align*}
     \hat{\xi}^{\Xi}_t \set & 
     \EE \left[ \frac{1}{\cosh (\tau^\kappa(t))} \Xi_T 
                        + \left(1-\frac{1}{\cosh (\tau^\kappa(t))}\right)  
     \int_t^T \xi_u K^{\Xi}(t,u) du \bigg\vert \cF_t \right],
   \end{align*}
   with the kernel
   $$
   K^{\Xi}(t,u) \set 
   \frac{\sinh(\tau^\kappa(u))}{\sqrt{\kappa}(\cosh(\tau^\kappa(t))-1)} , 
   \quad 0 \leq t \leq u < T.
   $$ 
   The solution $\hat{X}^\Xi$ of~\eqref{eq:SDEX2} satisfies the
   terminal constraint in the sense that
   $$
   \lim_{t \uparrow T} \hat{X}^{\Xi}_t = \Xi_T \quad \mathbb{P}\text{-a.s.}
   $$
   The minimal costs are given by
    \begin{align} 
      \inf_{u \in \cU^{\Xi}} J(u) = 
      & \frac{1}{2} \sqrt{\kappa} \coth(\tau^\kappa(0)) \left( x - \hat{\xi}^\Xi_0 \right)^2 +
        \frac{1}{2} \EE \left[ \int_0^T (\xi_t - \hat{\xi}^\Xi_t)^2 dt \right] \nonumber \\
      & + \frac{1}{2} \EE \left[ \int_0^T \sqrt{\kappa}
        \coth(\tau^\kappa(t)) d\langle \hat{\xi}^\Xi \rangle_t \right]
        < \infty. 
    \label{eq:cost2}
    \end{align}
\end{Theorem}

The convex-analytic proofs of Theorems \ref{thm:main1} and
\ref{thm:main2} are deferred to Section \ref{sec:proofs}.

Note that, rather than towards the \emph{current} target position
$\xi_t$, the optimal frictional hedging rules in (\ref{eq:SDEX1}) and
(\ref{eq:SDEX2}) prescribe to trade towards weighted averages
$\hat{\xi}_t$ and $\hat{\xi}_t^{\Xi}$, respectively, of \emph{expected
  future} target positions of $\xi$. Indeed, for each
$0 \leq t \leq T$, $K(t,.)$ and $K^{\Xi}(t,.)$ specify nonnegative
kernels which integrate to one over $[t,T]$, and so $\hat{\xi}$ and
$\hat{\xi}^{\Xi}$ average out the expected future positions of
$\xi$. For $\hat{\xi}^{\Xi}$ one chooses a convex combination of this
average of $\xi$ with the expected terminal position $\Xi_T$, where
the weight shifts gradually to $\Xi_T$ as $t \uparrow T$ since
$1/\cosh(\tau^k(t)) \uparrow 1$ in that case.

According to \eqref{eq:SDEX1} and \eqref{eq:SDEX2}, the optimal
tracking rate trades towards these targets at a speed proportional to
their distance to the investor's position at any time. The coefficient
of proportionality is controlled by both the cost parameter $\kappa$
and the remaining time-to-maturity $T-t$. For the unconstrained
solution in (\ref{eq:SDEX1}), since
$\lim_{t \uparrow T} \tanh (\tau^{\kappa}(t)) = 0$, trading slows down
when approaching the final time $T$; in other words, towards the end,
the investor does not worry about tracking $\xi$ anymore, but seeks to
minimize trading costs. This becomes intuitive when comparing the
effect of early interventions to later ones: with early interventions
the investor ensures that she stays reasonably close to the target for
the foreseeable future, but late interventions only can impact the
investor's performance for very short periods and therefore do not
warrant, at least asymptotically, the associated costs.  For the
constrained solution in (\ref{eq:SDEX2}) by contrast, we have
$\lim_{t \uparrow T} \coth (\tau^{\kappa}(t)) = +\infty$ and so the
optimal strategy trades with increased urgency towards
$\hat{\xi}^{\Xi}$, which itself is easily seen to converge to the
ultimate target position
$\Xi_T = \lim_{t \uparrow T} \hat{\xi}^{\Xi}_{t}$
$\PP$-a.s. (cf. Proof of Theorem \ref{thm:main2} below in Section
\ref{sec:proofs}).

Our tracking result generalizes an observation from
\citet{GarlPeder:13.2} from their homogeneous Markovian optimal
investment problem to a general hedging problem with a general
predictable target strategy $\xi$, also allowing for a random terminal
portfolio position $\Xi_T$.  It also sheds further light on the
general structure of optimal portfolio strategies in markets with
frictions. Indeed, the description of (asymptotically) optimal trading
strategies obtained in~\citet{MorMKSoner:15}, \citet{KallMK:14},
or~\citet{GuasoniWeb:15_1},~\cite{GuasoniWeb:15_2} prescribe a
reversion towards the frictionless strategy $\xi$ itself, \emph{not}
towards an average such as $\hat{\xi}$ or $\hat{\xi}^{\Xi}$.  For
sufficiently smooth $\xi$, e.g., of diffusion type, this is still
optimal asymptotically for small liquidity costs as then these
averages do not differ significantly from $\xi$. The next section,
however, shows that this is no longer the case when we allow for
singularities in the reference strategy.

Finally, our representations \eqref{eq:cost1} and \eqref{eq:cost2} for
the values of the tracking problems \eqref{eq:optProb1} and
\eqref{eq:optProb2}, respectively, show how these depend on the
initial position $x$ and the $L^2$-distance between the target $\xi$
and the respective signal processes $\hat{\xi}$ and
$\hat{\xi}^\Xi$. It also reveals the importance of the signals'
quadratic variation $\langle \hat{\xi} \rangle$,
$\langle \hat{\xi}^\Xi \rangle$ which can be viewed as a measure for
how effectively one can predict the target positions $\xi$ and
$\Xi_T$. To the best of our knowledge, the key role played by the
signals $\hat{\xi}$, $\hat{\xi}^\Xi$ was not observed in the general
theory of stochastic linear-quadratic control problems as discussed,
e.g., by \citet{KohlmannTang:02}.

\begin{Remark}
  As mentioned in the description of our problem setup in
  Section~\ref{sec:problem}, the quadratic cost term in our objective
  function in \eqref{eq:objFun} is due to linear temporary price
  impact as in the model proposed by \citet{AlmgChr:01}. In this
  regard, one might likewise extend the objective functional also in
  order to account for expected costs resulting from linear
  \emph{permanent price impact} (cf. in \cite{AlmgChr:01}). This would
  lead to the inclusion of the additional term
  \begin{equation} \label{eq:permanent} \EE \left[ \theta \left(
        \int_0^T u_t dt \right)^2 \right] = \theta \EE\left[ (X^u_T -
      x)^2 \right]
  \end{equation}
  for some constant $\theta > 0$. For the constrained problem
  in~\eqref{eq:optProb2}, this extra cost term obviously does not
  depend on the strategy and is thus irrelevant. For the unconstrained
  problem in~\eqref{eq:optProb1}, these extra costs can be regarded as
  a penalization term forcing the final position $X^u_T$ to be close
  to the initial position $x$. For ease of exposition, we refrain in
  the present paper from inducing this additional term, since our main
  intention here is to consider to outline the key role played by the
  optimal tracking signals $\hat{\xi}$, $\hat{\xi}^{\Xi}$ in the
  description of the optimal control as well as the corresponding
  minimal costs. A more general setup allowing for stochastic price
  impact, stochastic volatility and a penalization on the terminal
  position as in~\eqref{eq:permanent} is left for future research.
\end{Remark}


\section{Illustrations}
\label{sec:illustrations}

In this section we present a few case studies illustrating the
structure of the optimal hedging strategies we found in Theorems
\ref{thm:main1} and \ref{thm:main2}. The first two case studies are
simple deterministic toy examples which allow us to understand the
effect of jumps as well as of initial and terminal positions. The
final case study considers a discretely monitored Asian option where
random jumps in the reference hedge occur naturally.

In the first two cases we assume the initial position to be $x=0$ and
consider a time horizon of $T=1$ when, in the constrained case, the
position has to be liquidated, i.e., $\Xi_T=0$.  We depict $\xi$ along
with its averages $\hat{\xi}$ and $\hat{\xi}^{\Xi}$, respectively, as
well as the corresponding optimal frictional hedges $\hat{X}$ and
$\hat{X}^{\Xi}$. We also include a ``myopic'' benchmark strategy
$\tilde{X}$ which directly targets $\xi$ (without final constraint)
given by
\begin{equation*}
  d\tilde{X}_t = \frac{1}{\sqrt{\kappa}} (\xi_t - \tilde{X}_t) dt , \quad 0 \leq t \leq T, 
\end{equation*}
in order to compare with analogous strategies considered in
\citet{RogerSin:10}, \citet{MorMKSoner:15}, \citet{GuasoniWeb:15_1},
\cite{GuasoniWeb:15_2}, and \citet{KallMK:14}.

\subsection{Frictionless deterministic hedge with a jump}

In our first case study we consider a deterministic target strategy
$\xi$ (solid blue line in Figure \ref{fig:ex1}) which can be viewed as
a stock-buying schedule that prescribes to hold one stock until time
$T/2$ when the position is doubled by a jump.
\begin{figure}[ht]
\centering
\includegraphics[scale=.75]{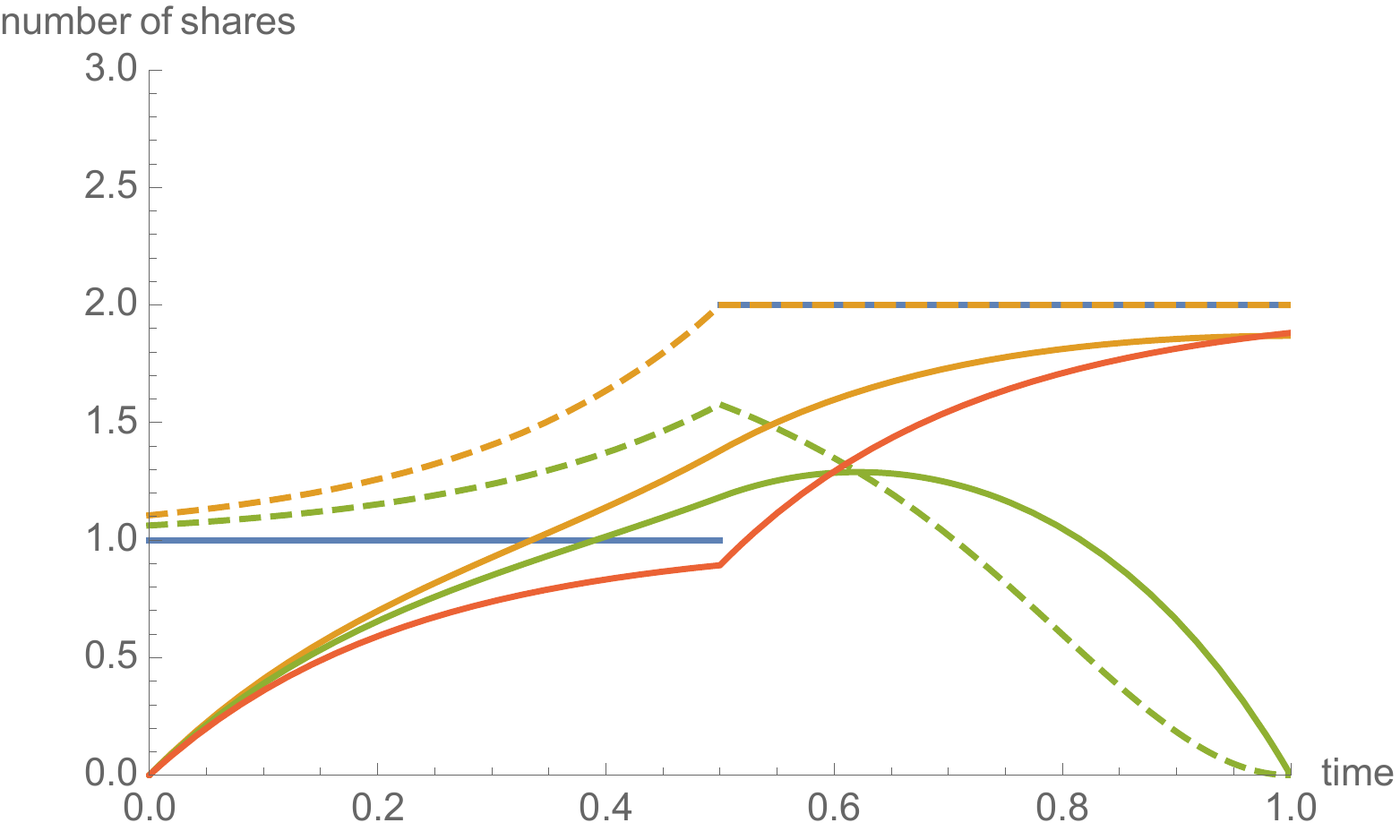}
\caption{Frictionless hedge $\xi$ with a jump at $t=T/2$ (blue),
  corresponding unconstrained (orange, dashed) and constrained (green,
  dashed) targets $\hat{\xi}$ and $\hat{\xi}^{\Xi}$, respectively, as
  well as the corresponding frictional hedges $\hat{X}$ (orange line)
  and $\hat{X}^{\Xi}$ (green line). The myopic benchmark hedge
  $\tilde{X}$ is plotted in red.}
\label{fig:ex1}
\end{figure}
One can observe that the effective target strategies $\hat{\xi}$ and
$\hat{\xi}^{\Xi}$ of the optimal controls $\hat{u}$ and
$\hat{u}^{\Xi}$, respectively, are smoothing out the jump of
$\xi$. The target $\hat{\xi}^{\Xi}$ additionally takes into account
the liquidation constraint $\Xi_T = 0$ of the agent's position until
maturity $T$. As expected, the optimal frictional hedges $\hat{X}$ and
$\hat{X}^{\Xi}$ are indeed anticipating the upward jump of the target
strategy $\xi$ at $t = T/2$ by building up their positions beyond the
actual current position of $\xi$ even before the occurrence of the
jump. This is not the case for the myopic benchmark strategy
$\tilde{X}$ which increases its position much more slowly and exhibits
a kink when the jump occurs after which trading speed picks up
significantly. Finally, the optimal holdings $\hat{X}^{\Xi}$ in the
constrained setting, where the position has to be unwound ultimately,
are decreasing when time approaches maturity and end up in the final
desired position $\hat{X}_T^{\Xi} = 0$.

\subsection{Frictionless deterministic hedge with a singularity}

The second target strategy $\xi$ (solid blue line in Figure
\ref{fig:ex2}) is again deterministic and also exhibits a singularity
midway at $t=T/2$, this time, however, it is a jump from $-\infty$ to
$+\infty$.
\begin{figure}[ht]
\centering
\includegraphics[scale=.75]{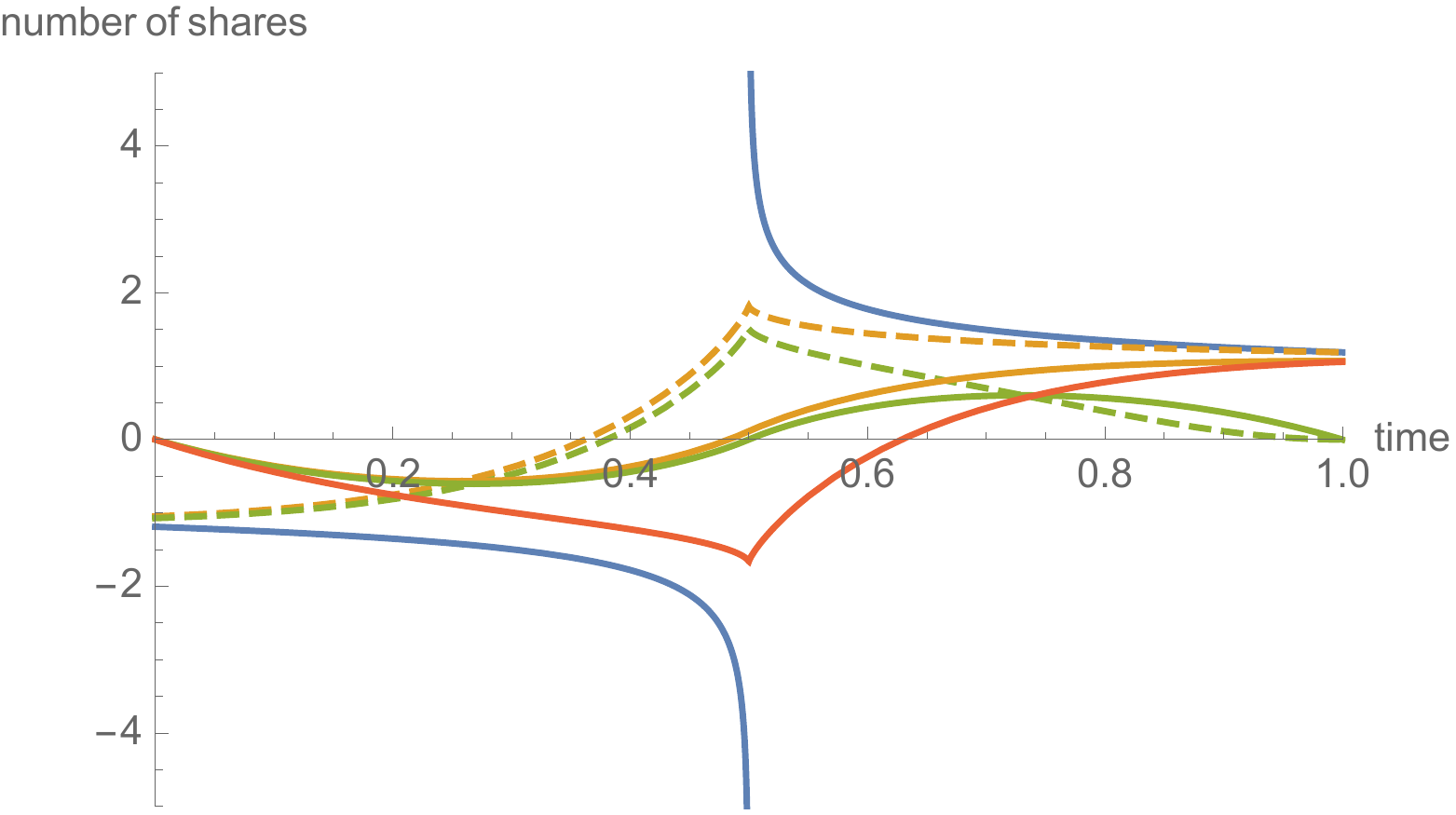}
\caption{Frictionless hedge $\xi$ with a singularity at $t=T/2$
  (blue), corresponding unconstrained (orange, dashed) and constrained
  (green, dashed) targets $\hat{\xi}$ and $\hat{\xi}^{\Xi}$,
  respectively, as well as the corresponding frictional hedges
  $\hat{X}$ (orange line) and $\hat{X}^{\Xi}$ (green line). The myopic
  benchmark hedge $\tilde{X}$ is plotted in red.}
\label{fig:ex2}
\end{figure}
Once more, one can observe that the effective target strategies
$\hat{\xi}$ and $\hat{\xi}^{\Xi}$ of the optimal controls $\hat{u}$
and $\hat{u}^{\Xi}$, respectively, are smoothing out the singularity
of $\xi$. Again, the target $\hat{\xi}^{\Xi}$ additionally takes into
account the liquidation constraint $\Xi_T = 0$ of the agent's position
until maturity $T$. In contrast to the benchmark strategy $\tilde{X}$,
the optimal frictional hedges $\hat{X}$ and $\hat{X}^{\Xi}$ are
anticipating the singularity of the target strategy $\xi$ at $t = T/2$
by gradually building up their positions before the singularity
occurs. Actually, they are trading \emph{away} from the current target
positions of $\xi$ for some time prior to $T/2$. This is in stark
contrast with the myopic benchmark strategy which keeps selling short
more and more intensely even milliseconds before the reference
strategy jumps to $+\infty$.

\subsection{Discrete Asian option}

In this final example we investigate a situation where the target
strategy $\xi$ is stochastic and exhibits a random jump. Specifically,
we consider hedging a discrete Asian call with maturity $T > 0$ in the
Bachelier model where the underlying risky asset $S$ is modeled by a
Brownian motion with volatility $\sigma > 0$:
\begin{equation*}
S_t = S_0 + \sigma W_t , \quad 0 \leq t \leq T.
\end{equation*}
For simplicity, we assume that the average is discretely monitored
over two fixing dates $T/2$ and $T$. That is, the payoff at maturity
$T$ is given by
$$H\set\left(\frac{1}{2}(S_{T/2} + S_T)-K\right)^+$$
for some strike $K \in \RR$.  The Bachelier price of the discrete
Asian option at time $t \in [0,T)$ can be computed as
\begin{align*}
  & \pi_t\set \\
  & \begin{cases}
    \sigma \sqrt{5 T/8 - t} \, \varphi\Big( \frac{S_t-K}{\sigma
      \sqrt{5T/8 - t}} \Big)+ S_t \Phi \Big(
    \frac{S_t-K}{\sigma\sqrt{5T/8 - t}} \Big) , & 0 \leq t < T/2 \\
    \frac{1}{2} \sigma \sqrt{T-t} \, \varphi\Big( \frac{S_{T/2} +
      S_t-2K}{\sigma\sqrt{T-t}} \Big)\\
    \qquad + \left(\frac{1}{2}
      (S_{T/2} + S_t)-K\right) \Phi \Big( \frac{S_{T/2} +
      S_t-2K}{\sigma\sqrt{T - t}} \Big) , & T/2 \leq t < T \\
\end{cases}
\end{align*}
where $\varphi$ and $\Phi$ denote the density and the cumulative
distribution function of the standard normal distribution,
respectively. Accordingly, the frictionless delta-hedging strategy is
$$
\xi_t = 
\begin{cases}
  \Phi \left( \frac{S_t-K}{\sigma\sqrt{5T/8 - t}} \right) , & 0 \leq t
  \leq T/2 \\
  \frac{1}{2} \Phi\left( \frac{S_{T/2} + S_t-2K}{\sigma \sqrt{T - t}}
  \right) , & T/2 < t < T.
\end{cases}
$$
Note that the delta-hedge exhibits a negative random jump at time $T/2$ since
\begin{equation*}
  \xi_{\frac{T}{2}+} - \xi_{\frac{T}{2}-} \set \lim_{t \downarrow
  \frac{T}{2}} \xi_t - \lim_{t \uparrow \frac{T}{2}} \xi_{t} =
-\frac{1}{2} \Phi\left( \frac{S_{T/2}-K}{\sigma \sqrt{T/8}} \right).
\end{equation*}
We assume that the initial position $x$ coincides with the initial
frictionless delta, i.e., e.g., $x=1/2$ in the case of an at-the-money
option with $K=S_0$. This allows us to focus on the hedging
performance itself and avoids distortions from the initial built up of
a sensible hedging position. As before, the terminal position will be
allowed to be either unconstrained or mandating liquidation, i.e.,
$\Xi_T=0$.

The effective targets $\hat{\xi}$ and $\hat{\xi}^{\Xi}$ of the optimal
frictional hedging strategy in (\ref{eq:SDEX1}) and (\ref{eq:SDEX2}),
respectively, can be explicitly computed:
\begin{equation*}
  \hat{\xi}_{t} =
  \begin{cases}
    \Phi\Big( \frac{2(S_t-K)}{\sigma\sqrt{5T/2 - 4t}} \Big) 
    \left( 1 - \frac{1}{2}
      \frac{\sinh(\tau^{\kappa}(T/2))}{\sinh(\tau^{\kappa}(t))}
    \right), & 0 \leq t < T/2, \\
    \xi_t, & T/2 \leq t < T, \\
  \end{cases}
\end{equation*}
and
\begin{equation*}
  \hat{\xi}^{\Xi}_{t} =
  \begin{cases}
    \Phi\Big( \frac{2(S_t-K)}{\sigma\sqrt{5T/2 - 4t}} \Big) \left( 1 -
      \frac{1}{2}
      \frac{\cosh(\tau^{\kappa}(T/2))+1}{\cosh(\tau^{\kappa}(t))} \right),
    & 0 \leq t < T/2 \\
    \left(1-\frac{1}{\cosh(\tau^{\kappa}(t))}\right)\xi_t, & T/2 \leq t <
    T. \\
  \end{cases}
\end{equation*}

Observe that the Bachelier delta-hedge $\xi$ is a martingale on
$[T/2,T]$ and thus the signal $\hat{\xi}$ coincides with it in thi
period. However, the optimal target $\hat{\xi}$ differs from the
frictionless hedge $\xi$ on $[0,T/2]$ since it is anticipating and
systematically smoothing out the random jump at $T/2$ whose size is
determined by the option's moneyness at this point.  The constrained
target $\hat{\xi}^{\Xi}$ anticipates the liquidation requirement at
maturity which plays a more and more dominating role after time $T/2$.

Again, the myopic benchmark strategy 
\begin{equation*}
  d\tilde{X}_t = \frac{\sigma}{\sqrt{\kappa}} (\xi_t - \tilde{X}_t) dt,
 \quad 0 \leq t <T  
\end{equation*}
is not taking into account the random jump at time $T/2$ and keeps on
tracking the frictionless delta-hedge even milliseconds before $T/2$
(see Figure \ref{fig:ex3}).

\begin{figure}[ht]
\centering
\includegraphics[scale=.6]{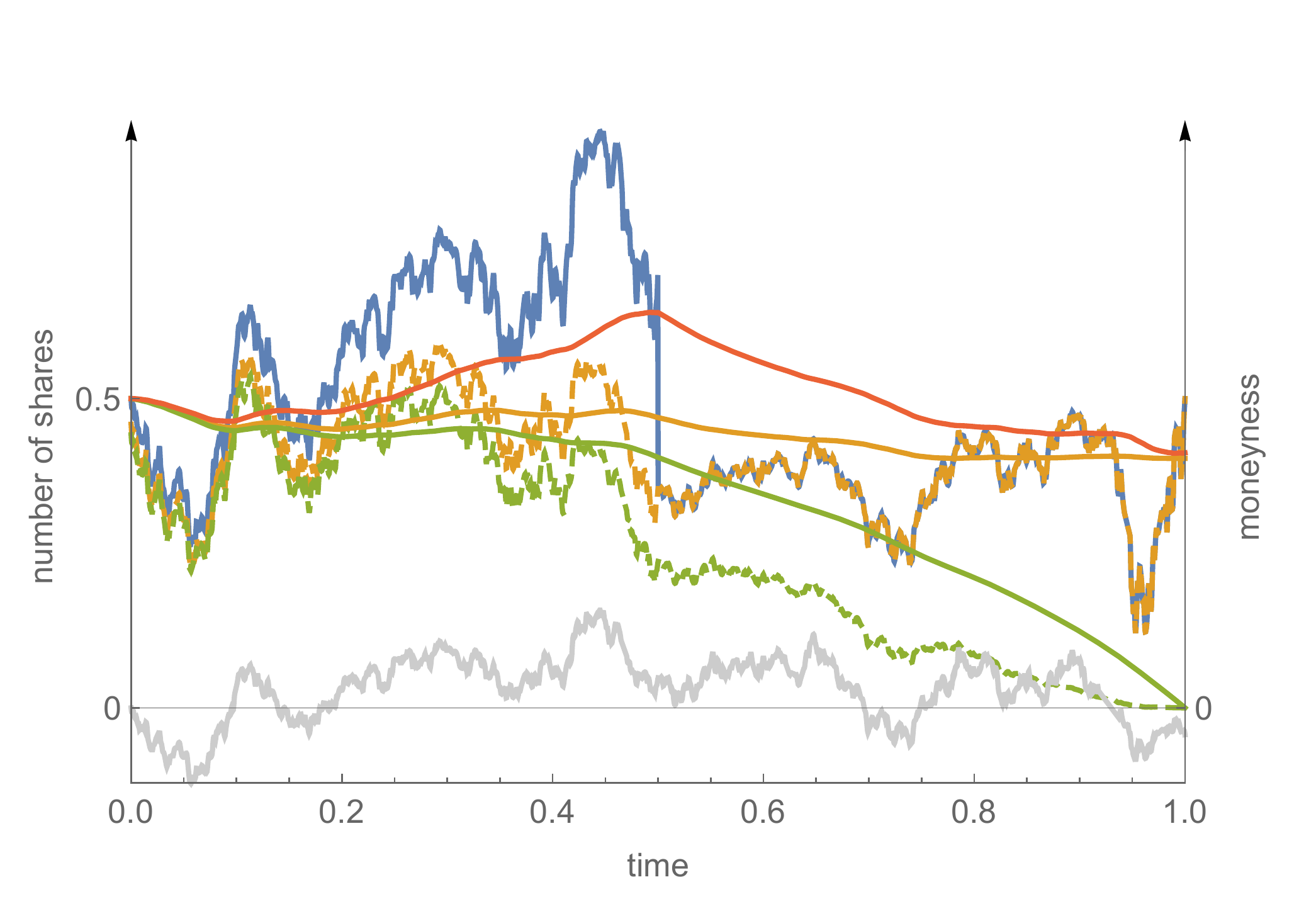}
\caption{
  Frictionless hedge $\xi$ with a jump at $t=T/2$ (blue),
  corresponding unconstrained (orange, dashed) and constrained (green,
  dashed) targets $\hat{\xi}$ and $\hat{\xi}^{\Xi}$, respectively,
  as well as the corresponding frictional hedges $\hat{X}$ (orange
  line) and $\hat{X}^{\Xi}$ (green line). The myopic benchmark hedge
  $\tilde{X}$ is plotted in red. The moneyness is indicated by the
  light gray line.}
\label{fig:ex3}
\end{figure}


\section{Proofs}
\label{sec:proofs}

In order to prove our main Theorems \ref{thm:main1} and
\ref{thm:main2} we use tools from convex analysis. Note that the
performance functional $u \mapsto J(u)$ in (\ref{eq:objFun}) is
strictly convex. Given a control $u \in \cU$ recall the definition of
the G\^ateaux derivative of $J$ at $u$ in the direction of
$w \in L^2(\PP \otimes dt)$:
\begin{equation*}
\langle J'(u), w \rangle \set \lim_{\rho \rightarrow 0} \frac{J(u + \rho w)-J(u)}{\rho}.
\end{equation*}
The following lemma provides an explicit expression for the G\^ateaux
derivative of our performance functional $J$:

\begin{Lemma}
\label{lem:gateaux}
For $u \in \cU$ we have
\begin{equation*}
   \langle J'(u), w \rangle = 
   \EE \left[ \int_0^T w_s \left( \kappa u_s + 
   \int_s^T (X^u_t - \xi_t) dt \right) ds \right]
\end{equation*}
for any $w \in L^2(\PP \otimes dt)$.
\end{Lemma}

\begin{proof}
Let $\rho > 0$, $u \in \cU$ and $w \in L^2(\PP \otimes dt)$. Note that
$X_t^{u+\rho w} = X^u_t + \rho \int_0^t w_s ds$. Then, we have
\begin{align*}
   J(u + \rho w)-J(u) = & \, \rho \EE \left[ \int_0^T \kappa u_t w_t + 
   \left(\int_0^t w_s ds \right) (X^u_t - \xi_t) dt \right] \\
   & + \rho^2 \EE \left[ \frac{\kappa}{2} \int_0^T w^2_t dt + 
   \frac{1}{2} \int_0^T \left(\int_0^t w_s ds \right)^2 dt \right].
\end{align*}
Hence,
\begin{equation*}
   \langle J'(u), w \rangle = 
   \EE \left[ \int_0^T \kappa u_t w_t + 
   \left(\int_0^t w_s ds \right) (X^u_t - \xi_t) dt \right].
\end{equation*}
Note that due to Fubini's Theorem we can write the second part of the
above integral as
\begin{equation*}
   \int_0^T \left(\int_0^t w_s ds \right) (X^u_t - \xi_t) dt =
   \int_0^T \left( \int_s^T (X^u_t - \xi_t) dt \right) w_s ds
\end{equation*}
which finally yields the assertion.
\end{proof}

Let us next derive necessary and sufficient first order conditions for
problems (\ref{eq:optProb1}) and (\ref{eq:optProb2}).

\begin{Lemma}[First order conditions] 
   \label{lem:FOC}
   \begin{enumerate}
      \item In the unconstrained problem (\ref{eq:optProb1}), a control 
      $\hat{u} \in \cU$ with $X\set X^{\hat{u}}$ minimizes the functional 
      $J$ if and only if $X$ satisfies
      \begin{equation} 
      \label{eq:FBSDE1}
         X_0 = x, \quad d\dot{X}_t = \frac{1}{\kappa} (X_t - \xi_t) dt + dM_t 
         \text{ for } 0 \leq t \leq T, \quad \dot{X}_T = 0,
      \end{equation}
      for a suitable square integrable martingale $(M_t)_{0 \leq t \leq T}$.
      \item In the constrained problem (\ref{eq:optProb2}), a control 
      $\hat{u} \in \cU^{\Xi}_x$ with $X\set X^{\hat{u}}$ minimizes the functional 
      $J$ if and only if $X$ satisfies
      \begin{equation} 
      \label{eq:FBSDE2}
         X_0 = x, \quad d\dot{X}_t = \frac{1}{\kappa} (X_t - \xi_t) dt + dM_t 
         \text{ for } 0 \leq t < T, \quad X_T = \Xi_T,
      \end{equation}
      for a suitable square integrable martingale $(M_t)_{0 \leq t < T}$.
\end{enumerate}
\end{Lemma}

In other words, the first order conditions in (\ref{eq:FBSDE1}) and
(\ref{eq:FBSDE2}) are taking the form of a coupled linear forward
backward stochastic differential equation (FBSDE) for the pair
$(X,u)$:
\begin{align*}
   dX_t & = u_t dt, \\
   du_t & = \frac{1}{\kappa} (X_t - \xi_t) dt + dM_t,
\end{align*}
with some square integrable martingale $M$ subject to
\begin{equation*}
   X_0=x \text{ and } 
   \begin{cases}
   u_T = 0 & \text{unconstrained case,}\\
   X_T=\Xi_T &\text{constrained case.}
   \end{cases}
\end{equation*}

\begin{proof} \emph{1.)}
We start with the unconstrained problem (\ref{eq:optProb1}). Since we
are minimizing the strictly convex functional $u \mapsto J(u)$ over
$\cU$, a necessary and sufficient condition for the optimality of
$\hat{u} \in \cU$ with corresponding $X^{\hat{u}} = x + \int_0^{\cdot}
\hat{u}_s ds$ is given by
\begin{equation*}
   \langle J'(\hat{u}), w \rangle = 0 
   \text{ for all } w \in \cU
\end{equation*}
(cf., e.g., \citet{EkelTem:99}). In view of Lemma~\ref{lem:gateaux}
this means that $\hat{u} \in \cU$ is optimal if and only if
\begin{equation} \label{eq:FOC1}
   \EE \left[ \int_0^T w_s \left( \kappa \hat{u}_s +  
   \int_s^T (X^{\hat{u}}_t - \xi_t) dt \right) ds \right] = 0
\end{equation}
for all $w \in \cU$. We will now show that the first order condition
in (\ref{eq:FOC1}) is satisfied (i.e., $\hat{u} \in \cU$ is optimal)
if and only if $X^{\hat{u}}$ satisfies the dynamics in
(\ref{eq:FBSDE1}).
 
\emph{Necessity:} Assume that $\hat{u} \in \cU$ with
$X^{\hat{u}} = x + \int_0^{\cdot} \hat{u}_s ds$ minimizes $J$, i.e., 
condition (\ref{eq:FOC1}) is satisfied by
$\hat{u}$. Then, by 
Fubini's Theorem and optional projection, we also get that
$$
\EE \left[ \int_0^T w_s \left( \kappa \hat{u}_s + 
\EE \left[\int_s^T (X^{\hat{u}}_t - \xi_t) dt \bigg\vert \cF_s\right]
\right) ds \right] = 0 
$$
for all $w \in \cU$. However, this is only possible if
\begin{equation} \label{eq:sol1FBSDE1}
   \hat{u}_s = - \frac{1}{\kappa} \EE \left[ \int_s^T (X^{\hat{u}}_t -
     \xi_t) dt \bigg\vert \cF_s \right]
   \quad d\PP \otimes ds \text{-a.e. on } \Omega \times [0,T].
\end{equation}
Hence, by defining the square integrable martingale 
\begin{equation} \label{eq:martFBSDE1}
M_s \set \EE \left[ \int_0^T (X^{\hat{u}}_t - \xi_t) dt \bigg\vert
  \cF_s \right], \quad 0 \leq s \leq T,
\end{equation}
we obtain the representation
\begin{equation} \label{eq:sol2FBSDE1}
   \hat{u}_s = -\frac{1}{\kappa} \left( M_s - \int_0^s (X^{\hat{u}}_t
     - \xi_t) dt \right) \quad d\PP \otimes ds \text{-a.e. on } \Omega
   \times [0,T],
\end{equation}
in other words, $X^{\hat{u}}$ satisfies the dynamics in
(\ref{eq:FBSDE1}). In particular, $X^{\hat{u}}_0 = x$ and
$\dot{X}_T^{\hat{u}} =\hat{u}_T = 0$ $\PP$-a.s.

\emph{Sufficiency:} Assume now that $\hat{u} \in \cU$ with
corresponding $X^{\hat{u}}$ satisfies the dynamics in
(\ref{eq:FBSDE1}) with $X^{\hat{u}}_0 = x$ and $\dot{X}_T^{\hat{u}} =
0$ $\PP$-a.s. Note that the unique strong solution to this linear
FBSDE in (\ref{eq:FBSDE1}) is indeed given by (\ref{eq:sol1FBSDE1})
or, equivalently, by (\ref{eq:sol2FBSDE1}). However, using this
representation of $\hat{u}$ and applying Fubini's Theorem yields
\begin{align*}
   &\EE \left[ \int_0^T w_s \left( \kappa \hat{u}_s +  
   \int_s^T (X^{\hat{u}}_t - \xi_t) dt \right) ds \right] 
   = \EE \left[ \int_0^T w_s \left( M_T - M_s \right) ds  \right] \\
   &= \EE \left[ \int_0^T w_s \EE[ M_T - M_s \vert \cF_s] ds  \right]
     = \int_0^T \EE \left[  w_s \left( \EE[M_T \vert \cF_s] - M_s
     \right) \right] ds
   = 0  
\end{align*}
for all $w \in \cU$, since $M$ is a martingale. Consequently, the
first order condition in (\ref{eq:FOC1}) is satisfied and $\hat{u} \in
\cU$ is optimal.

\emph{2.)} Similar as above, a necessary and sufficient condition  for
the optimality of $\hat{u}^{\Xi} \in \cU^{\Xi}_x$ with corresponding
$X^{\hat{u}^{\Xi}} = x + \int_0^{\cdot} \hat{u}^{\Xi}_s ds$ satisfying
$X^{\hat{u}^{\Xi}}_T = \Xi_T$ $\PP$-a.s. for the constrained problem
(\ref{eq:optProb2}) is given by
\begin{equation*}
   \langle J'(\hat{u}^{\Xi}), w \rangle = 0 
   \text{ for all } w \in \cU^0_0.
\end{equation*}
In contrast to the unconstrained case, observe now that we have an
additional constraint on $w$. Again, in view of
Lemma~\ref{lem:gateaux}, we get that $\hat{u}^{\Xi}\in \cU^{\Xi}_x$ is
optimal if and only if
\begin{equation} \label{eq:FOC2}
   \EE \left[ \int_0^T w_s \left( \kappa \hat{u}^{\Xi}_s +  
   \int_s^T (X^{\hat{u}^{\Xi}}_t - \xi_t) dt \right) ds \right] = 0 
   \text{ for all } w \in \cU^0_0.
\end{equation}

We will now show that the first order condition in 
(\ref{eq:FOC2}) 
is fulfilled (i.e., $\hat{u}^{\Xi} \in \cU^{\Xi}_x$ is optimal) 
if and only if $X^{\hat{u}^{\Xi}}$ satisfies the dynamics in
(\ref{eq:FBSDE2}).
 
\emph{Sufficiency:} Assume that $\hat{u}^{\Xi} \in \cU^{\Xi}_x$ with
corresponding $X^{\hat{u}^{\Xi}}$ satisfies the dynamics in
(\ref{eq:FBSDE2}) with $X_0^{\hat{u}^{\Xi}} = x$ and
$X_T^{\hat{u}^{\Xi}} = \Xi_T$ $\PP$-a.s. That is, we have the
representation
\begin{equation*}
  \hat{u}^{\Xi}_t = \hat{u}^{\Xi}_0 + \frac{1}{\kappa} \int_0^t
  (X^{\hat{u}^{\Xi}}_s - \xi_s) ds + M_t \quad d\PP \otimes dt
  \text{-a.e. on } \Omega \times [0,T)
\end{equation*}
for some square integrable martingale $(M_t)_{0 \leq t < T}$. From
$\hat{u}^{\Xi},\xi \in L^2(\PP \otimes dt)$, it follows that
$\EE[\int_0^T M_s^2 ds] < \infty$. Defining the square
integrable martingale
 \begin{equation*}
   N^{\Xi}_s \set \EE \left[ \int_0^T (X^{ \hat{u}^{\Xi}}_t - \xi_t)
     dt \bigg\vert \cF_s \right], \quad 0 \leq s \leq T, 
\end{equation*}
the above representation of $\hat{u}^{\Xi}$ yields
\begin{align*}
  &\EE \left[ \int_0^T w_s \left( \kappa \hat{u}^{\Xi}_s +  
   \int_s^T (X^{\hat{u}^{\Xi}}_t - \xi_t) dt \right) ds \right] \\
  & = \EE \left[ \int_0^T w_s \left( \kappa \hat{u}^{\Xi}_0 +
     N^{\Xi}_T + \kappa M_s \right) ds  \right] \\
  & = \EE \left[ ( \kappa \hat{u}^{\Xi}_0 +
     N^{\Xi}_T) \int_0^T w_s ds \right] + \kappa \EE \left[ \int_0^T
    w_s M_s ds \right] \\ 
  & = 0 \text{ for all } w \in  \cU^0_0
\end{align*}
by virtue of Lemma~\ref{lem:aux} below. Consequently, the first order
condition in (\ref{eq:FOC2}) is satisfied and $\hat{u}^{\Xi} \in
\cU^{\Xi}_x$ is optimal.

\emph{Necessity:} As shown in the proof of Theorem \ref{thm:main2}
below (which does \emph{not} use the necessity assertion of the
present Lemma), the optimal control $\hat{u}^\Xi$ in \eqref{eq:SDEX2}
satisfies the dynamics in \eqref{eq:FBSDE2}. Moreover, by strict
concavity of the objective functional in \eqref{eq:objFun}, the
solution to problem \eqref{eq:optProb2} is unique. Therefore, the
assertion is indeed necessary.
\end{proof}

The following technical Lemma is needed in the proof of
Lemma~\ref{lem:FOC} for the constrained problem \eqref{eq:optProb1}. 
 
\begin{Lemma}
\label{lem:aux}
Let $M$ be an adapted c\`adl\`ag process on $[0,T)$ with $\EE[\int_0^T
M_s^2 ds] < \infty$. Then, 
\begin{equation}
\label{eq:lemaux}
\EE\left[\int_0^T w_s M_s ds\right] = 0 \text{ for all } 
w \in \cU^0_0
\end{equation}
if and only if $M$ is a square integrable martingale on $[0,T)$.
\end{Lemma}

\begin{proof}
  First, assume that $M$ is a square integrable martingale on $[0,T)$
  with $\EE[\int_0^T M_s^2 ds] < \infty$.  Consider a $w \in \cU^0_0$
  such that $w = 0$ on $\Omega \times [T-\epsilon,T]$ for some
  $\epsilon > 0$. Then, by applying Fubini's Theorem we have
\begin{equation*}
  \EE \left[ \int_0^T w_s M_s ds \right] = \EE \left[
    \int_0^{T-\epsilon} w_s \EE[M_{T-\epsilon} \vert \cF_s] ds \right]
  = \EE \left[ M_{T-\epsilon} \int_0^T w_s ds \right] = 0.
\end{equation*}
Now, let $w \in \cU^0_0$ be arbitrary and consider an approximating
sequence $(w^{(n)})_{n \geq 1} \subset \cU^0_0$ with $w^{(n)} = 0$ on
$\Omega \times [T-\epsilon_n,T]$ for some $\epsilon_n \downarrow 0$
such that $w^{(n)} \rightarrow w$ in $L^2(\Omega \times [0,T],\PP
\otimes dt)$ for $n \rightarrow \infty$. Then, by the Cauchy-Schwarz
inequality we obtain
\begin{equation*}
\lim_{n \rightarrow \infty} \EE \left[ \int_0^T \vert (w^{(n)}_s - w_s) M_s \vert ds\right] = 0. 
\end{equation*}
Consequently, 
\begin{equation*}
  \EE \left[ \int_0^T w_s M_s ds \right] 
  = \lim_{n \rightarrow \infty} \EE \left[ \int_0^T w^{(n)}_s M_s ds \right] = 0,
\end{equation*}
where the last identity follows from our initial consideration for
$w$s with support in $[T-\epsilon,T]$, $\epsilon>0$. Hence, the
condition in (\ref{eq:lemaux}) is satisfied.

Conversely, assume now that the condition in (\ref{eq:lemaux}) is
satisfied. We have to show that $M$ is a square integrable martingale
on $[0,T)$.  Let $0 \leq t < u < T$, $A \in \cF_t$, be
arbitrary. For any $\epsilon > 0$ such that $t+\epsilon,u+\epsilon < T$ we define
\begin{equation*}
  w^{\epsilon}_s(\omega) \set 1_A(\omega) \frac{1}{\epsilon} \left( 1_{[t,t+\epsilon]}(s) - 
    1_{[u,u+\epsilon]}(s) \right) \text{ on } \Omega \times [0,T].
\end{equation*}
Obviously, $w$ is progressively measurable, in $L^2(\PP \otimes ds)$
and satisfies $\int_0^T w_s ds = 0$ $\PP$-a.s. Hence, by assumption
(\ref{eq:lemaux}) we have
\begin{align*}
  0 & = \EE\left [\int_0^T w^\epsilon_s M_s ds \right] = \EE \left[
      1_A \frac{1}{\epsilon} \int_t^{t+\epsilon} M_s ds \right] - \EE
      \left[ 1_A \frac{1}{\epsilon} \int_{u}^{u+\epsilon} M_s ds
      \right]. 
\end{align*}
Passing to the limit $\varepsilon \downarrow 0$, we obtain by
right-continuity of $M$,
\begin{align*}
  0 = \EE \left[ 1_A  (M_t  - M_u) \right] \text{ for all } 0 \leq t <
  u < T.
\end{align*}
Consequently, $M$ is a martingale on $[0,T)$. By assumption, we have
that $\EE[\int_0^T M_s^2 ds] < \infty$ which implies that $M$ is
square integrable on $[0,T)$.
\end{proof}

Now, we are ready to prove our main result by simple verification. We
start with Theorem \ref{thm:main1} for the unconstrained problem
\eqref{eq:optProb1}. 
\medskip

\noindent\textbf{Proof of Theorem~\ref{thm:main1}.} 
We divide the proof in two parts. First, we prove optimality of the
solution given in \eqref{eq:SDEX1}. Then, we compute the corresponding
minimal costs given in \eqref{eq:cost1}.

\emph{Optimality of \eqref{eq:SDEX1}:} In order to show that our
candidate in \eqref{eq:SDEX1} is the optimal solution for problem
\eqref{eq:optProb1}, we need to check the first order condition in
Lemma \ref{lem:FOC} 1.). For this, define the processes
\begin{equation*}
  Y_t \set \int_0^t \xi_s \cosh(\tau^\kappa(s)) ds \text{ and }
  \tilde{M}_t \set \EE [Y_T \vert \cF_t], \quad 0 \leq t \leq T.
\end{equation*}
Since $Y_T \in L^2(\PP)$, we have that $(\tilde{M}_t)_{0 \leq t \leq
  T}$ is a square integrable martingale. Moreover, note that
$Y,\tilde{M} \in L^2(\PP \otimes dt)$. Hence, the process $\hat{\xi}$
in Theorem \ref{thm:main1} can be written as
\begin{equation} \label{eq:xi1}
  \hat{\xi}_t = \frac{1}{\sqrt{\kappa} \sinh(\tau^\kappa(t))} \left(
    \tilde{M}_t - Y_t \right)  \quad d\PP \otimes dt\text{-a.e. on }
  \Omega \times [0,T)
\end{equation}
with corresponding dynamics
\begin{align} \label{eq:dynxi1}
  d\hat{\xi}_t = -\frac{\coth(\tau^\kappa(t))}{\sqrt{\kappa}} ( \xi_t
  - \hat{\xi}_t ) dt + \frac{1}{\sqrt{\kappa} \sinh(\tau^\kappa(t))}
  d\tilde{M}_t \quad \text{on } [0,T).
\end{align}
Due to Lemma \ref{lem:est} b), we know that $\hat{\xi} \in L^2(\PP
\otimes dt)$. Now, the density of the solution from \eqref{eq:SDEX1}
satisfies
\begin{align*}
  d\hat{u}_t = & - \frac{1}{\kappa} (1-\tanh(\tau^\kappa(t))^2) \left(
                 \hat{\xi}_t - \hat{X}_t \right) dt +
                 \frac{1}{\sqrt{\kappa}} \tanh(\tau^\kappa(t)) \left(
                 d\hat{\xi}_t - d\hat{X}_t \right) \\
   = & \frac{1}{\kappa} \left( \left( \hat{X}_t -\xi_t \right) dt +
       \frac{1}{\cosh(\tau^\kappa(t))} d\tilde{M}_t \right) \quad d\PP
       \otimes dt\text{-a.e. on } \Omega \times [0,T],
\end{align*}
that is, $\hat{u}$ satisfies the BSDE-dynamics in
(\ref{eq:FBSDE1}). Obviously, it holds that $\hat{X}_0 = x$. Solving
equation (\ref{eq:SDEX1}) for $\hat{X}$ yields upon differentiation
\begin{align}
  \hat{u}_t = &- \frac{1}{\sqrt{\kappa}}
                \frac{\sinh(\tau^\kappa(t))}{\cosh(\tau^\kappa(0))} x
                \nonumber \\
   & - \frac{1}{\kappa} \sinh(\tau^\kappa(t)) \int_0^t \hat{\xi}_s
     \frac{\sinh(\tau^\kappa(s))}{\cosh(\tau^\kappa(s))^2} ds +
     \frac{1}{\kappa} \frac{\tilde{M}_t -
     Y_t}{\cosh(\tau^\kappa(t))} \label{eq:u1}
\end{align}
and we observe that $\lim_{t \uparrow T} \hat{u}_t = 0$ $\PP$-a.s.,
i.e., the terminal condition in (\ref{eq:FBSDE1}) is indeed
satisfied. It remains to show that $\hat{u} \in L^2(\PP \otimes
dt)$. Since $\tilde{M}, Y \in L^2(\PP \otimes dt$), it suffices to
observe that $\sinh(\tau^\kappa(s))/\cosh(\tau^\kappa(s))^2$ is
bounded and therefore
\begin{align*}
\EE \left[ \int_0^T \left( \int_0^t \hat{\xi}_s
    \frac{\sinh(\tau^\kappa(s))}{\cosh(\tau^\kappa(s))^2} ds \right)^2
\right] dt
 & \leq \text{const} \, \EE \left[ \int_0^T \left( \int_0^t \vert
   \hat{\xi}_s \vert ds \right)^2 dt
\right] \\
& \leq \text{const} \, \frac{T^2}{2} \norm{\hat{\xi}}_{L^2(\PP \otimes
  dt)}^2 < \infty.
\end{align*}

\emph{Computation of minimal costs:} To compute the minimal costs
associated to the optimal control $\hat{u}$ given in
(\ref{eq:cost1}), note first that $\hat{u} \in
L^2(\PP\otimes dt)$ implies $\hat{X} \in L^2(\PP \otimes dt)$ and thus
$J(\hat{u}) < \infty$. For ease of presentation, we define
$$
c(t) \set \sqrt{\kappa} \tanh(\tau^\kappa(t)) , \quad 0 \leq t \leq T,
$$
so that $\hat{u}_t = c(t)(\hat{\xi}_t - \hat{X}_t)/\kappa$. Hence, the
minimal costs can be written as
\begin{align} 
  \infty > J(\hat{u}) = &~\EE \left[ \frac{1}{2} \int_0^T (\hat{X}_s - \xi_s)^2 ds 
   + \frac{1}{2} \kappa \int_0^T \hat{u}^2_s ds \right] \nonumber \\
  = & \lim_{t \uparrow T} \left\{ \frac{1}{2} \EE\left[ \int_0^t \hat{X}^2_s ds \right] 
   - \EE\left[ \int_0^t \hat{X}_s \xi_s ds \right] +
   \frac{1}{2} \EE\left[ \int_0^t \xi^2_s ds \right] \right. \nonumber
  \\
   & \hspace{30pt} + \frac{1}{2\kappa} \EE\left[ \int_0^t c(s)^2
     \hat{\xi}^2_s ds \right] - \frac{1}{\kappa} \EE\left[ \int_0^t
     c(s)^2 \hat{X}_s \hat{\xi}_s ds \right] \nonumber \\
   & \left. \hspace{30pt} + \frac{1}{2\kappa} \EE\left[ \int_0^t
     c(s)^2 \hat{X}^2_s ds \right] \right\},
  \label{eq1:p:cost1}
\end{align}
due to monotone convergence. Observe that, using integration
by parts and the dynamics of $\hat{\xi}$ from \eqref{eq:dynxi1}, we
have, for all $t < T$,
\begin{align*} 
  \EE[c(t) \hat{X}^2_t] = & c(0) x^2 
   + \frac{2}{\kappa} \EE\left[ \int_0^t c(s)^2 \hat{X}_s \hat{\xi}_s ds \right] \nonumber \\ 
   & - \frac{1}{\kappa} \EE\left[ \int_0^t c(s)^2 \hat{X}^2_s ds
     \right] - \EE\left[ \int_0^t \hat{X}^2_s ds \right]
\end{align*}
as well as
\begin{align*} 
  \EE[c(t) \hat{X}_t \hat{\xi}_t] = & c(0) \hat{\xi}_0 x 
   + \frac{1}{\kappa} \EE\left[ \int_0^t c(s)^2 \hat{\xi}^2_s ds \right] 
   - \EE\left[ \int_0^t \hat{X}_s \xi_s ds \right]
\end{align*}
and
\begin{align*} 
  \EE[c(t) \hat{\xi}^2_t] = & c(0) \hat{\xi}^2_0 
   + \frac{1}{\kappa} \EE\left[ \int_0^t c(s)^2 \hat{\xi}^2_s ds \right] 
   - 2 \EE\left[ \int_0^t \hat{\xi}_s \xi_s ds \right] \nonumber \\
   & +  \EE\left[ \int_0^t \hat{\xi}^2_s ds \right] +  \EE\left[
     \int_0^t c(s) d\langle \hat{\xi} \rangle_s \right]. 
\end{align*}
Using these identities, we can write \eqref{eq1:p:cost1} as
\begin{align} 
  \infty > J(\hat{u}) = & \lim_{t \uparrow T} \left\{ \frac{1}{2} c(0) (x - \hat{\xi}_0)^2 
    +  \frac{1}{2} \EE \left[ \int_0^t (\hat{\xi}_s - \xi_s)^2 ds
                          \right] \right. \nonumber \\
    & \left. \hspace{30pt} + \frac{1}{2} \EE \left[ \int_0^t c(s)
      d\langle \hat{\xi} \rangle_s \right]
    - \frac{1}{2} c(t)  \EE[ (\hat{X}_t - \hat{\xi}_t)^2]
      \right\}. \label{eq5:p:cost1}
\end{align}
To conclude our assertion for the minimal costs in $\eqref{eq:cost1}$,
observe that
$$
\EE[ (\hat{X}_{t} - \hat{\xi}_{t})^2] \leq 
2 \left( \EE[ \hat{X}_{t}^2] + \EE [\hat{\xi}_{t}^2] \right),
$$
and let us argue why
\begin{equation} \label{eq6:p:cost1}
  \lim_{t \uparrow T} c(t)  \EE[ \hat{X}_{t}^2] = 0
    \quad \text{and} \quad \lim_{t \uparrow T} c(t) \EE
    [\hat{\xi}_{t}^2] = 0.
\end{equation}
By Jensen's inequality, we have
$$
\EE[\hat{X}^2_t] \leq t \EE \left[ \int_0^t \hat{u}^2_s ds \right]
\leq T \, \EE \left[ \int_0^T \hat{u}^2_s ds \right] < \infty.
$$
Hence, due to $\lim_{t\uparrow T} c(t) = 0$, the first convergence in
\eqref{eq6:p:cost1} holds true. Concerning the second convergence in
\eqref{eq6:p:cost1}, we use the representation in \eqref{eq:xi1} for
$\hat{\xi}$ to obtain, again with Jensen's inequality as well as the
Cauchy-Schwarz inequality,
\begin{align*}
  0 \leq c(t) \EE[\hat{\xi}^2_t] & = \frac{c(t)}{\kappa
  \sinh(\tau^\kappa(t))^2} \EE[(\tilde{M}_t-Y_t)^2] \\
& \leq \frac{c(t)}{\kappa
  \sinh(\tau^\kappa(t))^2} \EE[(Y_T-Y_t)^2]  \\
& = \frac{c(t)}{\kappa
  \sinh(\tau^\kappa(t))^2} \EE\left[ \left( \int_t^T  \xi_s
  \cosh(\tau^\kappa(s))ds \right)^2 \right]  \\
  & \leq \frac{\cosh(\tau^\kappa(0))^2}{\sqrt{\kappa}
    \cosh(\tau^\kappa(t))} \frac{1}{\sinh(\tau^\kappa(t))}
   (T-t) \EE \left[ \int_t^T \xi^2_s ds \right] \\
  & \leq \frac{\cosh(\tau^\kappa(0))^2}{\cosh(\tau^\kappa(t))} \EE
    \left[ \int_t^T \xi^2_s ds \right] \underset{t \uparrow
    T}{\longrightarrow} 0, 
\end{align*}
where for the last inequality we used that $\sinh(\tau) \geq \tau$ for
all $\tau \geq 0$. In other words, also the second convergence in
\eqref{eq6:p:cost1} holds true. This finishes our proof of the
representation of the minimal costs in \eqref{eq:cost1}.
\qed   
\medskip

Next, we come to the proof of Theorem \ref{thm:main2} concerning the
constrained problem \eqref{eq:optProb2}.
\medskip

\noindent\textbf{Proof of Theorem~\ref{thm:main2}.}
Again, we will proceed in two steps. First, we prove optimality of the
solution given in \eqref{eq:SDEX2}. Then, we compute the corresponding
minimal costs given in \eqref{eq:cost2}.

\emph{Optimality of \eqref{eq:SDEX2}:} The verification of the
optimality of $\hat{X}^{\Xi} = x + \int_0^{\cdot} \hat{u}^{\Xi}_t dt$
in Theorem \ref{thm:main2} for the constrained problem
(\ref{eq:optProb2}) follows along the same lines as in the
unconstrained case. Again, we have to check the first order condition
in Lemma \ref{lem:FOC} 2.). For this, we define the processes
 
\begin{equation*}
  Y_t \set \frac{1}{\sqrt{\kappa}} \int_0^t \xi_s
  \sinh(\tau^\kappa(s)) ds \quad \text{and} \quad \tilde{M}^{\Xi}_t
  \set \EE [ Y_T + \Xi_T \vert \cF_t]
\end{equation*}
for all $0 \leq t \leq T$. Since $Z_T, \Xi \in L^2(\PP)$, we have that
$(\tilde{M}^{\Xi}_t)_{0 \leq t \leq T}$ is a square integrable
martingale. Moreover, note that $Y,\tilde{M}^\Xi \in L^2(\PP \otimes
dt)$. Hence, the process $\hat{\xi}^{\Xi}$ in Theorem~\ref{thm:main2}
can be written as
\begin{equation}
  \label{eq:xicon}
  \hat{\xi}^{\Xi}_t = \frac{1}{\cosh(\tau^\kappa(t))} \left(
    \tilde{M}^{\Xi}_t - Y_t \right) \quad d\PP \otimes
  dt\text{-a.e. on } \Omega \times [0,T]
\end{equation}
with corresponding dynamics
\begin{align} \label{eq:dynxi2}
  d\hat{\xi}^\Xi_t = -\frac{\tanh(\tau^\kappa(t))}{\sqrt{\kappa}} (
  \xi_t - \hat{\xi}^\Xi_t ) dt + \frac{1}{\cosh(\tau^\kappa(t))}
  d\tilde{M}^{\Xi}_t \quad \text{on } [0,T].
\end{align}
In particular, we observe that $\hat{\xi}^\Xi \in L^2(\PP \otimes
dt)$. Similar to the unconstrained case above, one easily checks that
\begin{equation*}
  d\hat{u}^{\Xi}_t = \frac{1}{\kappa} (\hat{X}^{\Xi}_t -
  \xi_t) dt + \frac{1}{\sqrt{\kappa}}
  \frac{1}{\sinh(\tau^\kappa(t))} d\tilde{M}^{\Xi}_t \quad d\PP
  \otimes dt\text{-a.e. on } \Omega \times [0,T),
\end{equation*}
that is, $\hat{u}^{\Xi}$ satisfies the dynamics in
(\ref{eq:FBSDE2}). Obviously, it holds that $\hat{X}^{\Xi}_0 = x$.
   
Next, we have to check the terminal condition in (\ref{eq:FBSDE2}),
that is, $\lim_{t \uparrow T} \hat{X}^{\Xi}_t = \Xi_T$ $\PP$-a.s. In
order to show this, first note that we can consider a c\`adl\`ag
version of $(\hat{\xi}_t^{\Xi})_{0 \leq t \leq T}$ due to its
representation in (\ref{eq:xicon}). Hence, since $\Xi_T$ is
$\cF_{T-}$-measurable by assumption \eqref{eq:ccp} we obtain the
$\PP$-a.s. limit
\begin{equation*}
  \lim_{t \uparrow T} \hat{\xi}^{\Xi}_t = \EE[ \Xi_T \vert \cF_{T-}] =
  \Xi_T
\end{equation*}
in (\ref{eq:xicon}). In other words, for every $\epsilon > 0$ there
exists a random time $\Upsilon_{\epsilon} \in [0, T)$ such that
$\PP$-a.s. 
\begin{equation*}
    \Xi_T - \epsilon \leq \hat{\xi}^{\Xi}_t \leq \Xi_T + \epsilon \quad
    \text{for all } t \in [\Upsilon_{\epsilon},T].
\end{equation*}
For $\lim_{t \uparrow T} \hat{X}^{\Xi}_t = \Xi_T$ $\PP$-a.s., it
clearly suffices to show that for any $\epsilon > 0$ it holds that
\begin{equation*}
  \limsup_{t \uparrow T} \hat{X}^{\Xi}_t \leq \Xi_T + \epsilon \quad
  \text{and} \quad 
  \liminf_{t \uparrow T} \hat{X}^{\Xi}_t \geq \Xi_T - \epsilon \quad
  \PP\text{-a.s.}
\end{equation*}
Define $X_t^{\epsilon} \set \Xi_T + \epsilon - \hat{X}^{\Xi}_t$ so
that $\hat{\xi}^{\Xi}_t - \hat{X}^{\Xi}_t \leq X_t^{\epsilon}$
$\PP$-a.s. for $t \in [\Upsilon_{\epsilon},T)$. This yields  
\begin{align*}
  dX^{\epsilon}_t = & -d\hat{X}^{\Xi}_t = -\frac{1}{\sqrt{\kappa}}
                      \coth(\tau^\kappa(t)) (\hat{\xi}^{\Xi}_t -
                      \hat{X}^{\Xi}_t ) dt \\
  \geq & -\frac{1}{\sqrt{\kappa}} \coth(\tau^\kappa(t)) X^{\epsilon}_t
         dt. 
\end{align*}
Moreover, note that for all $\omega \in \Omega$ the linear ODE on
$[\Upsilon_{\epsilon}(\omega),T)$ given by
\begin{equation*}
  dZ_t = -\frac{1}{\sqrt{\kappa}} \coth(\tau^\kappa(t)) Z_t dt, \quad
  Z_{\Upsilon_{\epsilon}(\omega)} =
  X^{\epsilon}_{\Upsilon_{\epsilon}(\omega)}(\omega),
\end{equation*}
admits the solution
\begin{equation*}
  Z_t =
  X^{\epsilon}_{\Upsilon_{\epsilon}}\exp\left(-\frac{1}{\sqrt{\kappa}}
    \int_{\Upsilon_{\epsilon}}^t \coth(\tau^\kappa(s)) ds \right) =
  X^{\epsilon}_{\Upsilon_{\epsilon}}
  \frac{\sinh(\tau^{\kappa}(t))}{\sinh(\tau^{\kappa}(\Upsilon_{\epsilon}))},
  \quad t < T,
\end{equation*}
with $\lim_{t \uparrow T} Z_t = 0$. By the comparison principle for
ODEs, we get $\PP$-a.s. $X^{\epsilon}_t \geq Z_t$  for all $t \in
[\Upsilon_{\epsilon},T)$. Hence,
\begin{equation*}
  \liminf_{t \uparrow T} X^{\epsilon}_t \geq \lim_{t \uparrow T} Z_t =
  0 \quad \PP\text{-a.s.},
\end{equation*}
that is, $\limsup_{t \uparrow T} \hat{X}^{\Xi}_t \leq \Xi_T +
\epsilon$ $\PP$-a.s. Similarly, define $\tilde{X}_t^{\epsilon} \set
\Xi_T - \epsilon - \hat{X}^{\Xi}_t$ and observe as above that
$\PP$-a.s. on $[\Upsilon_{\epsilon},T)$ we have  
\begin{equation*} 
  d\tilde{X}^{\epsilon}_t \leq -\frac{1}{\sqrt{\kappa}}
  \coth(\tau^\kappa(t)) \tilde{X}^{\epsilon}_t dt. 
\end{equation*}
Again, as above by comparison principle we obtain
\begin{equation*}
  \limsup_{t \uparrow T} \tilde{X}^{\epsilon}_t \leq 0 \quad \PP\text{-a.s.},
\end{equation*}
i.e., $\liminf_{t \uparrow T} \hat{X}^{\Xi}_t \geq \Xi_T - \epsilon$
$\PP$-a.s. as remained to be shown for \eqref{eq:FBSDE2}.

Finally, we have to argue that $\hat{u}^\Xi \in L^2(\PP \otimes
dt)$. For this, we may assume without loss of generality that
$x=0$. Moreover, let us denote $\hat{u}^{\Xi,\xi} \set \hat{u}^\Xi$,
$\hat{X}^{\Xi,\xi} \set \hat{X}^\Xi$ and $\hat{\xi}^{\Xi,\xi} \set
\hat{\xi}^{\Xi}$ to emphasize also the dependence on the given target
process $\xi$. With this notation it holds that
$$
\hat{u}^{\Xi,\xi} = \hat{u}^{\Xi,0} + \hat{u}^{0,\xi}.
$$
Hence, we have to show that $\hat{u}^{\Xi,0} \in L^2(\PP \otimes dt)$
and $\hat{u}^{0,\xi} \in L^2(\PP \otimes dt)$.
 
Concerning $\hat{u}^{\Xi,0}$, observe that, using $\hat{\xi}^{\Xi,0}_t
= \Xi_t/\cosh(\tau^\kappa(t))$ with $\Xi_t \set \EE[\Xi_T \vert
\cF_t]$, $0 \leq t \leq T$, as well as the explicit solution
$\hat{X}^{\Xi,0}_t$ for the ODE in \eqref{eq:SDEX2}, we obtain 
\begin{align}
  \hat{u}^{\Xi,0}_t = &
                        \frac{\coth(\tau^\kappa(t))}{\sqrt{\kappa}}
                        \left(\hat{\xi}_t^{\Xi,0}-\hat{X}^{\Xi,0}_t \right) \nonumber \\
  = & \frac{\coth(\tau^\kappa(t))}{\sqrt{\kappa}} \left( e^{-\int_0^t
      \frac{\coth(\tau^\kappa(u))}{\sqrt{\kappa}}du}
      \hat{\xi}^{\Xi,0}_0 + \right. \nonumber \\
                      & \hspace{68pt} \left. e^{-\int_0^t
                        \frac{\coth(\tau^\kappa(u))}{\sqrt{\kappa}}du} \int_0^t e^{\int_0^s
                        \frac{\coth(\tau^\kappa(u))}{\sqrt{\kappa}}du}
                        d\hat{\xi}^{\Xi,0}_s  \right) \nonumber \\
  = & \frac{\cosh(\tau^\kappa(t))}{\sqrt{\kappa}\sinh(\tau^\kappa(0))}
      \hat{\xi}^{\Xi,0}_0 +\frac{\cosh(\tau^\kappa(t))}{\kappa}
      \int_0^t \frac{\Xi_s}{\cosh(\tau^\kappa(s))^2} ds \nonumber \\
                      & + \frac{\cosh(\tau^\kappa(t))}{\sqrt{\kappa}} \int_0^t
                        \frac{2}{\sinh(2 \tau^\kappa(s))} d\Xi_s, 
                        \label{eq:p:ccp1}
\end{align}
where we used integration by parts in the second line. Obviously, the
first two terms in \eqref{eq:p:ccp1} belong to $L^2(\P \otimes
dt)$.
The third term is in $L^2(\PP \otimes dt)$ since, using Fubini's
Theorem as well as $\sinh(\tau) \geq \tau$ for all $\tau \geq 0$, we
get
  \begin{align*}
    & \EE \left[ \int_0^T \left(  \int_0^t  \frac{2 d\Xi_s }{\sinh(2 \tau^\kappa(s))} \right)^2 dt \right]
      = \EE \left[ \int_0^T \int_0^t \left(  \frac{2}{\sinh(2 \tau^\kappa(s))} \right)^2 
      d\langle \Xi \rangle_s dt \right] \\
    & = \EE \left[ \int_0^T (T-s) \left(  \frac{2}{\sinh(2 \tau^\kappa(s))} \right)^2 
      d\langle \Xi \rangle_s \right]
      \leq \EE \left[ \int_0^T \frac{\kappa}{T-s} d\langle \Xi \rangle_s \right] \\
    & = \kappa \int_0^T \frac{d\EE[\Xi^2_s]}{T-s} < \infty
  \end{align*}
by assumption \eqref{eq:ccp}.

Concerning $\hat{u}^{0,\xi}$, we use the explicit expressions for
$\hat{\xi}^{0,\xi}_t$ and $\hat{X}^{0,\xi}_t$ to obtain in
\eqref{eq:SDEX2} that
\begin{align}
  \hat{u}^{0,\xi}_t = 
  & \frac{\coth(\tau^\kappa(t))}{\sqrt{\kappa}}
    \left(\hat{\xi}_t^{0,\xi}-\hat{X}^{0,\xi}_t
    \right) \nonumber \\
  = & \frac{\cosh(\tau^\kappa(t))-1}{\sqrt{\kappa}\sinh(\tau^\kappa(t))}
      \EE\left[ \int_t^T \xi_u K^{\Xi}(t,u) du \Big\vert \cF_t \right] \nonumber \\
  & - \frac{\cosh(\tau^\kappa(t))}{\kappa} \int_0^t \frac{\cosh(\tau^\kappa(s))-1}{\sinh(\tau^\kappa(s))^2} 
    \EE\left[ \int_s^T\xi_u K^{\Xi}(s,u) du \Big\vert \cF_s\right] ds. \label{eq:p:ccp2}
\end{align}
Note that all the ratios in \eqref{eq:p:ccp2} involving the functions
$\cosh(\cdot)$ and $\sinh(\cdot)$ are actually bounded on
$[0,T]$. Moreover, we have by Lemma \ref{lem:est} c) below that
$$
\EE\left[ \int_t^T \xi_u K^{\Xi}(t,u) du \Big\vert \cF_t \right]
\in L^2(\PP \otimes dt),
$$
as well as, using Jensen's inequality,
\begin{align*}
  & \EE\left[ \int_0^T \left( \int_0^t  \EE\left[ \int_s^T\xi_u
    K^{\Xi}(s,u) du \Big\vert \cF_s\right] ds \right)^2 dt \right] \\
  & \leq \frac{T^2}{2} \EE\left[ \int_0^T \left( \EE\left[ \int_s^T\xi_u
    K^{\Xi}(s,u) du \Big\vert \cF_s\right] \right)^2 ds \right]
    < \infty.
\end{align*}
Together, this shows $\hat{u}^{\Xi} \in L^2(\PP \otimes dt)$ as
desired. 
  
\emph{Computation of minimal costs:} Now, we compute the minimal costs
associated to the optimal control $\hat{u}^\Xi$ given in
\eqref{eq:cost2}. We will follow along the same lines as in the
unconstrained case above. First of all, note that
$\hat{u}^\Xi \in L^2(\PP\otimes dt)$ implies
$\hat{X}^\Xi \in L^2(\PP\otimes dt)$ and hence $J(\hat{u}) < \infty$.
For ease of presentation, we define
$$
c(t) \set \sqrt{\kappa} \coth(\tau^\kappa(t)), \quad 0 \leq t < T,
$$
i.e., $\hat{u}^\Xi_t=c(t)(\hat{\xi}^\Xi_t - \hat{X}^\Xi_t)/\kappa $.
Analogously to the unconstrained case above, we can write
$J(\hat{u}^\Xi)$ as
\begin{align}
  \infty > J(\hat{u}^\Xi) = 
  & \lim_{t \uparrow T} \left\{ \frac{1}{2} c(0) (x - \hat{\xi}^\Xi_0)^2 
    +  \frac{1}{2} \EE \left[ \int_0^t (\hat{\xi}^\Xi_s - \xi_s)^2 ds \right] \right. \nonumber \\
  & \left. + \frac{1}{2} \EE \left[ \int_0^t c(s) d\langle \hat{\xi}^\Xi \rangle_s \right]
    - \frac{1}{2} c(t)  \EE[ (\hat{X}^\Xi_t - \hat{\xi}^\Xi_t)^2] \right\}. \label{eq1:p:cost2}
\end{align}
To conclude our assertion for the minimal costs in
$\eqref{eq:cost2}$, observe that
$$
\EE[ (\hat{X}^\Xi_{t} - \hat{\xi}^\Xi_{t})^2] \leq 2 \left( \EE[
  (\hat{X}_{t}^\Xi - \Xi_t )^2] + \EE [(\Xi_t - \hat{\xi}^\Xi_{t})^2]
\right),
$$
where $\Xi_t \set \EE[\Xi_T \vert \cF_t]$, $0 \leq t \leq T$, and let
us argue why
\begin{equation} \label{eq2:p:cost2} \lim_{t \uparrow T} c(t) \EE[
  (\hat{X}_{t}^\Xi - \Xi_t )^2] = 0 \quad \text{and} \quad \lim_{t
    \uparrow T} c(t) \EE [(\Xi_t - \hat{\xi}^\Xi_{t})^2] = 0.
\end{equation}

Concerning the first convergence in \eqref{eq2:p:cost2}, Jensen's
inequality, monotonicity of the function $\cosh(\cdot)$ as well as the
estimate $\sinh(\tau) \geq \tau$ for all $\tau \geq 0$ yield
\begin{align}
  c(t) \EE[ (\hat{X}_{t}^\Xi - \Xi_t )^2] 
  & \leq c(t) \EE[ (\hat{X}_{t}^\Xi - \hat{X}^\Xi_T )^2] \nonumber \\ 
  & \leq \frac{\kappa \cosh(\tau^\kappa(0))}{T-t} \EE \left[ \left( \int_t^T \hat{u}^\Xi_s ds \right)^2 \right] 
    \nonumber \\
  & \leq \kappa \cosh(\tau^\kappa(0)) \EE \left[ \int_t^T (\hat{u}_s^\Xi)^2 ds \right] 
    \underset{t \uparrow T}{\longrightarrow} 0, \label{eq3:p:cost2}
\end{align}
since $\Xi_T = \hat{X}^\Xi_T$ and
$\hat{u}^\Xi \in L^2(\PP \otimes dt)$.

Concerning the second convergence in \eqref{eq2:p:cost2}, we
insert the definition for $\hat{\xi}^\Xi$ to obtain that
\begin{align*}
  & c(t) \EE[ (\Xi_t - \hat{\xi}^\Xi_{t})^2] \\
  & = c(t) \EE \left[ \left( \frac{\cosh(\tau^\kappa(t)) -
    1}{\cosh(\tau^\kappa(t))} \Xi_t \right. \right. \\
  & \hspace{60pt} -
   \left.\left. \frac{\cosh(\tau^\kappa(t)) - 1}{\cosh(\tau^\kappa(t))} \EE\left[
    \int_t^T \xi_u K^\Xi(t,u) du \Big\vert \cF_t
    \right] \right)^2 \right] \\
  & \leq 2 c(t) \left( \frac{\cosh(\tau^\kappa(t)) -
    1}{\cosh(\tau^\kappa(t))} \right)^2 \EE [\Xi^2_T] \\
  & \hspace{16pt} + 2 c(t) \left( \frac{\cosh(\tau^\kappa(t)) -
    1}{\cosh(\tau^\kappa(t))} \right)^2
    \EE\left[ 
    \int_t^T \xi^2_u K^\Xi(t,u) du \right] \\
  & \leq \frac{2\sqrt{\kappa}}{\cosh(\tau^\kappa(t))} \frac{(\cosh(\tau^\kappa(t)) - 1 )^2}{\sinh(\tau^\kappa(t))}
    \EE[\Xi^2_T] \\
  & \hspace{16pt} + \frac{2\sinh(\tau^\kappa(0))}{\cosh(\tau^\kappa(t))} \frac{\cosh(\tau^\kappa(t)) - 1}{\sinh(\tau^\kappa(t))} 
    \EE\left[ \int_t^T \xi^2_u du \right] \underset{t \uparrow T}{\longrightarrow} 0,
\end{align*}
since $\Xi_T \in L^2(\PP)$, $\xi \in L^2(\PP \otimes dt)$ and
$\lim_{t\uparrow T}(\cosh(\tau^\kappa(t)) - 1)/\sinh(\tau^\kappa(t)) =
0$.
Consequently, also the second convergence in \eqref{eq2:p:cost2} holds
true.  This finishes our proof of the representation of the minimal
costs in \eqref{eq:cost2}.  \qed \medskip

The next Lemma shows that the set $\cU^{\Xi}_x$ is not empty under the
assumption \eqref{eq:ccp}.

\begin{Lemma} \label{lem:ccp} For $\Xi_T \in L^2(\PP,\cF_T)$ we have
  that $\cU^{\Xi}_x \neq \varnothing$ if and only if condition
  \eqref{eq:ccp} holds, i.e., if and only if
  $\int_0^T \frac{d\EE[\Xi_t^2]}{T-t} < \infty$ with
  $\Xi_t \set \EE[ \Xi_T \vert \cF_t]$ for all $0 \leq t \leq T$.
\end{Lemma}

\begin{proof}
  Let $\Xi_T \in L^2(\PP,\cF_T)$. We first prove necessity. Assume
  there exists $u \in \cU^{\Xi}_x$, i.e., $u \in L^2(\PP \otimes dt)$
  such that
  $$
    X^u_T = x + \int_0^T u_s ds = \Xi_T.
  $$
  Then, applying Fubini's Theorem, we obtain
  $$
    \int_0^T\frac{d\EE[\Xi_t^2]}{T-t} = \frac{1}{T} (\EE[\Xi^2_T]- \EE[\Xi^2_0]) 
    + \int_0^T \EE[\Xi_T^2 - \Xi_s^2] d\left(\frac{1}{T-s}\right).
  $$
  Moreover,
  $\EE[\Xi_T^2 - \Xi_s^2] = \EE[(\Xi_T - \Xi_s)^2] \leq \EE[(X^u_T -
  X^u_s)^2]$
  due to the $L^2$-projection property of conditional
  expectations. Hence, we get
  \begin{align*}
    \int_0^T\frac{d\EE[\Xi_t^2]}{T-t} \leq & \frac{1}{T} (\EE[\Xi^2_T]- \EE[\Xi^2_0]) 
                                             + \int_0^T \EE \left[
                                             \left( \int_s^T u_r dr
                                             \right)^2 \right] d\left(\frac{1}{T-s}\right) \\
    = &  \frac{1}{T} (\EE[\Xi^2_T]- \EE[\Xi^2_0]) 
           + \EE \left[ \int_0^T \left( \frac{1}{T-s} \int_s^T u_r dr \right)^2 ds \right] 
           < \infty
  \end{align*}
  by $\Xi_T \in L^2(\PP)$ and Lemma \ref{lem:est} a).
  
  For sufficiency, simply consider the optimizer $\hat{u}^{\Xi}$ from
  Theorem \ref{thm:main2} which we proved to be in $\cU^\Xi_x$ under
  the condition \eqref{eq:ccp}.
\end{proof}

The final Lemma collects estimates concerning the
$L^2(\PP \otimes dt)$-norm which are needed several times in the
proofs above.

\begin{Lemma} \label{lem:est} 
  Let $(\zeta_t)_{0 \leq t \leq T} \in L^2(\PP \otimes dt)$ be
  progressively measurable. Moreover, let $K(t,u)$, $K^\Xi(t,u)$,
  $0 \leq t \leq u < T$, denote the kernels from Theorems
  \ref{thm:main1} and \ref{thm:main2}, respectively.
  \begin{itemize}
  \item[a)] For
    $\bar{\zeta}_t \set \frac{1}{T-t} \int_t^T \zeta_s ds$, $t < T$,
    we have
  $$\norm{\bar\zeta}_{L^2(\PP \otimes dt)} \leq 2 \norm{\zeta}_{L^2(\PP \otimes dt)}.$$
\item[b)] For
  $\zeta^K_t \set \EE[\int_t^T \zeta_u K(t,u) du \vert \cF_t]$,
  $t < T$, we have
  $$\norm{\zeta^K}_{L^2(\PP \otimes dt)} \leq c \norm{\zeta}_{L^2(\PP \otimes dt)}$$
  for some constant $c > 0$.
\item[c)] For
  $\zeta^{K^{\Xi}}_t \set \EE[\int_t^T \zeta_u K^{\Xi}(t,u) du \vert
  \cF_t]$, $t < T$, we have
  $$\norm{\zeta^{K^{\Xi}}}_{L^2(\PP \otimes dt)} \leq c \norm{\zeta}_{L^2(\PP \otimes dt)}$$
  for some constant $c>0$.
  \end{itemize} 
\end{Lemma}

\begin{proof}
a) By Fubini's Theorem and the Cauchy-Schwarz inequality, we have
\begin{align*}
  \norm{\bar\zeta}^2_{L^2(\PP\otimes dt)} 
  & = \EE \left[ \int_0^T \int_0^T
    \zeta_r \zeta_s \int_0^{r \wedge s} \left(\frac{1}{T-t}\right)^2 dt
    dr ds \right] \\
 & = \EE \left[ \int_0^T \int_0^T
    \zeta_r \zeta_s \frac{1}{T-r \wedge s} dr ds \right]  - \frac{1}{T} \EE\left[ \left(
    \int_0^T \zeta_s ds \right)^2 \right]\\
  & \leq \EE \left[ 2 \int_0^T \zeta_r \int_0^r
    \zeta_s \frac{1}{T-s} ds dr\right] \\
  & = 2 \EE \left[ \int_0^T \zeta_s \left( \frac{1}{T-s} \int_s^T
    \zeta_r dr \right) ds \right] \\
  & \leq 2 \norm{\zeta}_{L^2(\PP \otimes dt)}
    \norm{\bar\zeta}_{L^2(\PP \otimes dt)}
\end{align*}
and hence the assertion.

b) First, assume that $(\zeta_t)_{0 \leq t \leq T}$ is deterministic,
and so $\zeta^K_t = \int_t^T \zeta_u K(t,u) du$. By similar
computations as in a) we obtain
\begin{align*}
  \norm{\zeta^K}^2_{L^2(dt)} 
  & = \int_0^T \int_0^T
    \zeta_r \zeta_s \int_0^{r \wedge s} K(t,r)K(t,s) dt
    dr ds \\
  & \leq \int_0^T \int_0^T
    \zeta_r \zeta_s \frac{1}{\sqrt{\kappa}} \cosh(\tau^\kappa(r)) \cosh(\tau^\kappa(s))
    \coth(\tau^\kappa(r \wedge s)) dr ds \\
  & = 2 \int_0^T \zeta_r \frac{\cosh(\tau^\kappa(r))}{\sqrt{\kappa}} \int_0^r
    \zeta_s \cosh(\tau^\kappa(s)) \coth(\tau^\kappa(s)) ds dr \\
  & = 2 \int_0^T \zeta_s \cosh(\tau^\kappa(s))^2 \zeta^{K}_s ds \\
  & \leq 2 \cosh(\tau^\kappa(0))^2 \norm{\zeta}_{L^2(dt)} \norm{\zeta^K}_{L^2(dt)},
\end{align*}
i.e., $\norm{\zeta^K}_{L^2(dt)} \leq c \norm{\zeta}_{L^2(dt)}$ for
some constant $c>0$.  Now, for general
$(\zeta_t)_{0 \leq t \leq T} \in L^2(\PP \otimes dt)$ progressively
measurable, we get with Fubini's Theorem
$$
\EE \left[ \int_0^T (\zeta^K_t)^2 dt \right] = 
\int_0^T \int_t^T \int_t^T \EE\Big[ \EE[\zeta_r\vert\cF_t]
\EE[\zeta_s\vert\cF_t] \Big] K(t,r) K(t,s) dr ds dt.
$$
Again, application of Cauchy-Schwarz's and Jensen's inequalities
yields
$$
\EE\left[ \EE[\zeta_r\vert\cF_t] \EE[\zeta_s\vert\cF_t] \right] \leq
\norm{\zeta_r}_{L^2(\PP)} \norm{\zeta_s}_{L^2(\PP)}, \quad t \leq r,s
\leq T.
$$
Consequently,
\begin{align*}
  \norm{\zeta^K}_{L^2(\PP \otimes dt)}^2 
  & \leq 
    \int_0^T \int_t^T \int_t^T \norm{\zeta_r}_{L^2(\PP)}
    \norm{\zeta_s}_{L^2(\PP)} K(t,r) K(t,s) dr ds dt \\
  & =\int_0^T \left( \int_t^T \norm{\zeta_r}_{L^2(\PP)} K(t,r) dr \right)^2
    dt.
\end{align*}
Now, put $\tilde{\zeta}_t \set \norm{\zeta_t}_{L^2(\PP)}$ and apply
the estimate already proved for deterministic functions to conclude
\begin{align*}
  \norm{\zeta^K}_{L^2(\PP \otimes dt)}^2 
  & =\int_0^T \left( \int_t^T \tilde{\zeta}_r K(t,r) dr \right)^2
    dt \\
& \leq c \int_0^T \vert \tilde{\zeta}_t \vert^2 dt
= c \int_0^T \EE[\zeta_t^2] dt = c \norm{\zeta}^2_{L^2(\PP \otimes dt)}.
\end{align*}

c) Jensen's inequality and Fubini's Theorem give
\begin{align*}
  \norm{\zeta^{K^{\Xi}}}^2_{L^2(\PP \otimes dt)} 
  & = \EE \left[ \int_0^T (\zeta^{K^\Xi}_t)^2 dt \right]
    \leq \int_0^T \int_t^T \EE[\zeta_u^2] K^\Xi(t,u) du dt \\                                      
  & = \int_0^T \EE[\zeta_u^2] \int_0^u K^\Xi(t,u) dt du.
\end{align*}
Now, using $\cosh(\tau)-1 \geq \tau^2/2$ for all $\tau \geq 0$, we get
\begin{align*}
  0 & \leq \int_0^u K^\Xi(t,u) dt = \int_0^u
      \frac{\sinh(\tau^\kappa(u))}{\sqrt{\kappa}(\cosh(\tau^\kappa(t))-1)}
      dt \\
    & \leq
      \frac{\sinh(\tau^\kappa(u))}{\sqrt{\kappa}}
      \int^u_0\frac{2\kappa}{(T-t)^2} dt
      \leq
      2 \sqrt{\kappa} \, \frac{\sinh(\tau^\kappa(u))}{T-u} 
      \underset{u \uparrow T}{\longrightarrow} 1.
\end{align*}
Thus, the above integral over $K^\Xi$ is bounded uniformly in 
$0 \leq u \leq T$  by some constant $c >0$, and so
$$
\norm{\zeta^{K^{\Xi}}}^2_{L^2(\PP \otimes dt)}
\leq c \int_0^T \EE[\zeta_u^2] du = c \, \norm{\zeta}_{L^2(\PP \otimes dt)}^2
$$
yielding the assertion in c).
\end{proof}


\bibliographystyle{plainnat} \bibliography{finance}

\begin{thebibliography}{25}
\providecommand{\natexlab}[1]{#1}
\providecommand{\url}[1]{\texttt{#1}}
\expandafter\ifx\csname urlstyle\endcsname\relax
  \providecommand{\doi}[1]{doi: #1}\else
  \providecommand{\doi}{doi: \begingroup \urlstyle{rm}\Url}\fi

\bibitem[Alfonsi et~al.(2010)Alfonsi, Fruth, and Schied]{AlfFruSch:10}
Aur\'elien Alfonsi, Antje Fruth, and Alexander Schied.
\newblock Optimal execution strategies in limit order books with general shape
  functions.
\newblock \emph{Quantitative Finance}, 10\penalty0 (2):\penalty0 143--157,
  2010.
\newblock \doi{10.1080/14697680802595700}.
\newblock URL \url{http://dx.doi.org/10.1080/14697680802595700}.

\bibitem[Almgren and Chriss(2001)]{AlmgChr:01}
Robert Almgren and Neil Chriss.
\newblock Optimal execution of portfolio transactions.
\newblock \emph{J. Risk}, 3:\penalty0 5--39, 2001.

\bibitem[Almgren and Li(2015)]{AlmgLi:15}
Robert Almgren and Tianhui~Michael Li.
\newblock Option hedging with smooth market impact.
\newblock Preprint, June 2015.

\bibitem[Almgren(2003)]{Almg:03}
Robert~F. Almgren.
\newblock Optimal execution with nonlinear impact functions and
  trading-enhanced risk.
\newblock \emph{Applied Mathematical Finance}, 10\penalty0 (1):\penalty0 1--18,
  2003.
\newblock \doi{10.1080/135048602100056}.
\newblock URL \url{http://dx.doi.org/10.1080/135048602100056}.

\bibitem[Black and Scholes(1973)]{BlacSch:73}
Fischer Black and Myron Scholes.
\newblock The pricing of options and corporate liabilities.
\newblock \emph{Journal of Political Economy}, 81:\penalty0 637--654, May--June
  1973.

\bibitem[Cai et~al.(2015)Cai, Rosenbaum, and Tankov]{CaiRosenbaumTankov:15}
Jiatu Cai, Mathieu Rosenbaum, and Peter Tankov.
\newblock Asymptotic lower bounds for optimal tracking: a linear programming
  approach.
\newblock Preprint, October 2015.

\bibitem[Cartea and Jaimungal(2015)]{CarJai:15}
\'{A}lvaro Cartea and Sebastian Jaimungal.
\newblock A closed-form execution strategy to target {VWAP}.
\newblock Preprint, January 2015.

\bibitem[Ekeland and T\'emam(1999)]{EkelTem:99}
I.~Ekeland and R.~T\'emam.
\newblock \emph{Convex Analysis and Variational Problems}.
\newblock Society for Industrial and Applied Mathematics, 1999.
\newblock \doi{10.1137/1.9781611971088}.
\newblock URL \url{http://epubs.siam.org/doi/abs/10.1137/1.9781611971088}.

\bibitem[F\"ollmer and Sondermann(1986)]{FollSond:86}
H.~F\"ollmer and D.~Sondermann.
\newblock Hedging of non-redundant contingent claims.
\newblock In \emph{Contributions to mathematical economics}, pages 205--223.
  North-Holland, Amsterdam-New York-Oxford-Tokyo, 1986.

\bibitem[Frei and Westray(2013)]{FreiWes:13}
Christoph Frei and Nicholas Westray.
\newblock Optimal execution of a {VWAP} order: A stochastic control approach.
\newblock \emph{Mathematical Finance}, 2013.
\newblock \doi{10.1111/mafi.12048}.

\bibitem[G\^arleanu and Pedersen(2013{\natexlab{a}})]{GarlPeder:13.1}
Nicolae G\^arleanu and Lasse~Heje Pedersen.
\newblock Dynamic trading with predictable returns and transaction costs.
\newblock \emph{The Journal of Finance}, 68\penalty0 (6):\penalty0 2309--2340,
  2013{\natexlab{a}}.
\newblock ISSN 1540-6261.
\newblock \doi{10.1111/jofi.12080}.
\newblock URL \url{http://dx.doi.org/10.1111/jofi.12080}.

\bibitem[G\^arleanu and Pedersen(2013{\natexlab{b}})]{GarlPeder:13.2}
Nicolae G\^arleanu and Lasse~Heje Pedersen.
\newblock Dynamic portfolio choice with frictions.
\newblock Preprint, May 2013{\natexlab{b}}.

\bibitem[G\"{o}kay et~al.(2011)G\"{o}kay, Roch, and Soner]{GokayRochSoner:11}
Selim G\"{o}kay, Alexandre~F. Roch, and H.~Mete Soner.
\newblock Liquidity models in continuous and discrete time.
\newblock In Giulia Di~Nunno and Bernt {\O}ksendal, editors, \emph{Advanced
  Mathematical Methods for Finance}, pages 333--365. Springer Berlin
  Heidelberg, 2011.
\newblock ISBN 978-3-642-18412-3.
\newblock URL \url{http://dx.doi.org/10.1007/978-3-642-18412-3_13}.

\bibitem[Guasoni and Weber(2015{\natexlab{a}})]{GuasoniWeb:15_1}
Paolo Guasoni and Marko Weber.
\newblock Dynamic trading volume.
\newblock \emph{Mathematical Finance}, pages n/a--n/a, 2015{\natexlab{a}}.
\newblock ISSN 1467-9965.
\newblock \doi{10.1111/mafi.12099}.
\newblock URL \url{http://dx.doi.org/10.1111/mafi.12099}.

\bibitem[Guasoni and Weber(2015{\natexlab{b}})]{GuasoniWeb:15_2}
Paolo Guasoni and Marko Weber.
\newblock Nonlinear price impact and portfolio choice.
\newblock Preprint, June 2015{\natexlab{b}}.

\bibitem[Gu\'eant and Pu(2015)]{GueantPu:15}
Olivier Gu\'eant and Jiang Pu.
\newblock Option pricing and hedging with execution costs and market impact.
\newblock Preprint, April 2015.
\newblock URL \url{http://arxiv.org/abs/1311.4342}.

\bibitem[Kallsen and Muhle-Karbe(2014)]{KallMK:14}
Jan Kallsen and Johannes Muhle-Karbe.
\newblock High-resilience limits of block-shaped order books.
\newblock Preprint, September 2014.

\bibitem[Kohlmann and Tang(2002)]{KohlmannTang:02}
Michael Kohlmann and Shanjian Tang.
\newblock Global adapted solution of one-dimensional backward stochastic
  {R}iccati equations, with application to the mean-variance hedging.
\newblock \emph{Stochastic Processes and their Applications}, 97\penalty0
  (2):\penalty0 255 -- 288, 2002.
\newblock ISSN 0304-4149.
\newblock \doi{http://dx.doi.org/10.1016/S0304-4149(01)00133-8}.

\bibitem[Merton(1973)]{Mert:73}
Robert~C. Merton.
\newblock Theory of rational option pricing.
\newblock \emph{Bell J. Econom. and Management Sci.}, 4:\penalty0 141--183,
  1973.
\newblock ISSN 0741-6261.

\bibitem[Moreau et~al.(2015)Moreau, Muhle-Karbe, and Soner]{MorMKSoner:15}
Ludovic Moreau, Johannes Muhle-Karbe, and H.~Mete Soner.
\newblock Trading with small price impact.
\newblock \emph{Mathematical Finance}, pages n/a--n/a, 2015.
\newblock ISSN 1467-9965.
\newblock \doi{10.1111/mafi.12098}.
\newblock URL \url{http://dx.doi.org/10.1111/mafi.12098}.

\bibitem[Naujokat and Westray(2011)]{NaujWes:11}
Felix Naujokat and Nicholas Westray.
\newblock Curve following in illiquid markets.
\newblock \emph{Mathematics and Financial Economics}, 4\penalty0 (4):\penalty0
  299--335, 2011.
\newblock ISSN 1862-9679.
\newblock \doi{10.1007/s11579-011-0042-5}.
\newblock URL \url{http://dx.doi.org/10.1007/s11579-011-0042-5}.

\bibitem[Obizhaeva and Wang(2013)]{ObiWang:13}
Anna~A. Obizhaeva and Jiang Wang.
\newblock Optimal trading strategy and supply/demand dynamics.
\newblock \emph{Journal of Financial Markets}, 16\penalty0 (1):\penalty0 1--32,
  2013.
\newblock ISSN 1386-4181.
\newblock \doi{http://dx.doi.org/10.1016/j.finmar.2012.09.001}.
\newblock URL
  \url{http://www.sciencedirect.com/science/article/pii/S1386418112000328}.

\bibitem[Predoiu et~al.(2011)Predoiu, Shaikhet, and Shreve]{PredShaShr:11}
Silviu Predoiu, Gennady Shaikhet, and Steven Shreve.
\newblock Optimal execution in a general one-sided limit-order book.
\newblock \emph{SIAM Journal on Financial Mathematics}, 2\penalty0
  (1):\penalty0 183--212, 2011.
\newblock \doi{10.1137/10078534X}.
\newblock URL \url{http://dx.doi.org/10.113710078534X}.

\bibitem[Rogers and Singh(2010)]{RogerSin:10}
L.~C.~G. Rogers and Surbjeet Singh.
\newblock The cost of illiquidity and its effects on hedging.
\newblock \emph{Mathematical Finance}, 20:\penalty0 597--615, 2010.

\bibitem[Schied and Sch{\"o}neborn(2009)]{SchSchon:09}
Alexander Schied and Torsten Sch{\"o}neborn.
\newblock Risk aversion and the dynamics of optimal liquidation strategies in
  illiquid markets.
\newblock \emph{Finance Stoch.}, 13\penalty0 (2):\penalty0 181--204, 2009.
\newblock ISSN 0949-2984.

\end{thebibliography}

\end{document}